\documentclass[12pt,intlimits,reqno]{amsart}

\usepackage{amsmath}

\usepackage{enumerate}

\usepackage{amssymb}

\usepackage{amsthm}

\usepackage{bbm}

\usepackage{graphicx}

\usepackage{mathrsfs}
\usepackage[all]{xy}

\usepackage[polish,english]{babel}
\usepackage{verbatim}
\usepackage[utf8]{inputenc}

\newtheorem{thm}{Theorem}[section]

\newtheorem{prop}[thm]{Proposition}

\theoremstyle{definition}

\newtheorem{example}[thm]{Example}

\theoremstyle{remark}

\numberwithin{equation}{section}

\renewcommand{\Re}{\mathrm{Re}}

\renewcommand{\Im}{\mathrm{Im}}

\begin{document}

\title{Classical and quantum Kummer shape algebras}



\maketitle

\begin{center}
A. Odzijewicz, E. Wawreniuk\\
Institute of Mathematics\\
University in Bia\l{}ystok\\
Cio\l{}kowskiego 1M, 15-245 Bia\l{}ystok, Poland\\
aodzijew@uwb.edu.pl , ewawreniuk@math.uwb.edu.pl
\end{center}

\tableofcontents

\section*{Abstract}

We study a family of integrable systems of nonlinearly coupled harmonic oscillators on the classical and quantum levels. We show that the integrability of these systems follows from their symmetry characterized by algebras called here Kummer shape algebras. The resolution of identity for a wide class of reproducing kernels is found. A number of examples, illustrating this theory, is also presented.

\section{Introduction}

It is a commonly accepted point of view that  systems of nonlinearly coupled harmonic oscillators provide good models for the description of a wide class of physical phenomena. First of all, they appear naturally in optics where light modes interact between themselves through the medium in which polarization depends on the electromagnetic field in a nonlinear way \cite{Perina}, \cite{Mil}. On the other hand, one uses them for modeling the dynamics of various classical mechanical systems, see for example \cite{holm}. The oscillations of some modes in molecular dynamics could also be described in this manner \cite{Allen}.

In the Kummer's paper \cite{Kummer} a hamiltonian system, in which the interaction between harmonic oscillators is described by some homogeneous polynomial, was integrated. The method of integration is based on the reduction of the system to a hamiltonian system on $\mathbb{R}^3$ with Poisson-Nambu bracket defined by a circularly symmetric function $\mathcal{C}$. In this way one obtains the trajectories of the reduced system as intersections of the $0$-level surface $\mathcal{C}^{-1}(0)$ of $\mathcal{C}$, called in \cite{holm} the Kummer shape, with the level sets of the reduced Hamiltonian. Next, using this system, the KAM-theory is applied for the extraction of rigorous information for a class of more general Hamiltonians.

In this paper we use Manley-Rowe integrals of motion for integration of the above type of hamiltonian systems. As a result we obtain a Poisson subalgebra of the standard Poisson algebra $C^\infty (\mathbb{R}^{2N+2})$, which reduces to the Poisson algebra on the Kummer shape, see \eqref{casimir}, \eqref{nambu}). We call this reduced Poisson algebra the classical Kummer shape algebra. 

The above method of integration works also when we pass to the corresponding quantum system of nonlinearly coupled harmonic oscillators. As a consequence, we obtain an operator algebra generated by three operators given by (\ref{ra0}-\ref{rastar}) in Section \ref{sec:qkummer}. These operators satisfy the relations (\ref{rel1}-\ref{rel2}) which depend on the structural function $\mathcal{G}_\hbar $ defined by the interaction part of the considered Hamiltonian \eqref{qH}. This algebra, called here a quantum Kummer shape algebra, describes the symmetry of the considered quantum system and, in the limit $\hbar \to 0$, corresponds to the classical Kummer shape algebra. We show this correspondence in Section \ref{sec:4} using the $*_\hbar$-product defined by the coherent state map of the reduced system. We will also show that this correspondence intertwines the classical and quantum reductions of the system, see Fig.2 presented at the end of Section \ref{sec:4}. The resolution of identity for the reduced coherent states is found by combining classical and quantum reduction procedures, see Proposition \ref{prop4} in Section \ref{sec:7}.

The quantum Kummer shape algebras were investigated in a systematic way in \cite{ATer}, \cite{AO}, \cite{AO2}, where the authors had in mind their applications to nonlinear phenomena in quantum optics. We believe it was V.P. Karassiov, see \cite{Kar1}, \cite{Kar2}, who first applied these structures to quantum optics problems.  They were also discussed in numerous papers by different authors as the so-called deformed bosonic oscillators algebras, e.g. see \cite{Jag1}, \cite{Jag2} and references therein.

The connection of these algebras with the coherent state method of quantization and the theory of $q$-special functions was investigated in \cite{AO1}.

In Section \ref{sec:integrability} we discuss the integrability of the quantum Hamiltonian \eqref{qH} and the correspondence of the Heisenberg equations defined by \eqref{qH} to the Hamilton equations (\ref{i_0t}-\ref{dyt}) if $\hbar \to 0$.

In Section \ref{sec:examples} several examples of Kummer shapes are presented which show their importance for various questions of mathematical physics.

\section{Classical Kummer shape algebra}\label{sec:kummer}

We will investigate a hamiltonian system consisting of non-linearly coupled $N+1$ harmonic oscillators. The phase space of such a system is $\mathbb{C}^{N+1}$ equipped with the standard Poisson bracket 

\begin{equation}\label{pb}
\{ f,g\}= -i \sum_{n=0}^N\left(\frac{\partial f}{\partial z_n}\frac{\partial g}{\partial \bar z_n}- \frac{\partial g}{\partial z_n}\frac{\partial f}{\partial \bar z_n}\right) ,
\end{equation}
of $f,g\in C^\infty(\mathbb{C}^{N+1})$. So, for the coordinate functions we have 
\begin{equation}\label{zety}
\{z_i,\bar z_j\}=-i\delta_{ij},\quad \{z_i,z_j\}=0,  \textrm{ and } \{\bar z_i,\bar z_j\}=0.
\end{equation}
We will take 
\begin{multline}\label{H} 
H=h_0(|z_0|^2,|z_1|^2,\ldots, |z_N|^2)+g_0(|z_0|^2, |z_1|^2,\ldots, |z_N|^2) z_0^{l_0} z_1^{l_1}\cdots z_N^{l_N}+\\
+g_0(|z_0|^2, |z_1|^2,\ldots, |z_N|^2)z_0^{-l_0} z_1^{-l_1}\cdots z_N^{-l_N}
\end{multline}
as a Hamiltonian of the system, where the following convention is used
\begin{equation}\label{zet}
 z^{l_i}_i=\left\{\begin{array}{ll}

z_i^{l_i} & \textrm{if }l_i\geq 0  \\ 

 \bar z^{-l_i}_i & \textrm{if }l_i<0

                      \end{array}\right. , 
\end{equation}
$z_i\in\mathbb{C}$ and $l_i\in\mathbb{Z}$. We will also assume that $h_0$ and $g_0$ are smooth functions on $\mathbb{C}^{N+1} \cong \mathbb{R}^{2N+2}$ and that $g_0$ is positive valued. Additionally, without loss of generality, we exclude the case when all exponents $l_0, \ldots, l_N$ are negative.

In our investigation we will follow Kummer's paper \cite{Kummer}, where the Hamiltonian \eqref{H}, with the coupling function $g_0$ as a real positive constant and $h_0$ as a polynomial of degree smaller than $|l_0|+ \ldots + |l_N|$, was considered. 

In order to integrate the system given by Hamiltonian \eqref{H} we pass to the new canonical coordinates $(I_0, \ldots , I_N, \psi_0, \ldots , \psi_N ) $ defined by 
 \begin{equation}\label{int}
 I_k:= \sum_{j=0}^N \rho_{kj} |z_j|^2,
 \end{equation}
 \begin{equation}\label{psil}
\psi_l := \sum_{j=0}^N \kappa_{jl}\phi_j,
\end{equation}
 where the real $(N+1)\times (N+1)$ matrix $\rho= (\rho_{ij})$ satisfies
\begin{equation}\label{omega}
\det \rho \neq 0 \quad \mbox{and} \quad \sum_{j=0}^N \rho_{ij}l_j = \delta_{0i} 
\end{equation}
and $\phi_j$ is the argument of $z_j=|z_j|e^{i\phi_j}$. The matrix $\kappa= (\kappa_{ij}) $ is the inverse of $\rho =(\rho_{ij})$.

In the sequel, we will restrict our considerations to the open subset $\Omega^{N+1} := \{ (z_0, \ldots, z_N )^T \in \mathbb{C}^{N+1} : |z_k| > 0 , \mbox{ for } k=0,1,\ldots, N \}$ of $\mathbb{C}^{N+1}$, which is invariant with respect to the hamiltonian flows 
\begin{equation}\label{flow}
\sigma_r (t) (z_0, \ldots, z_N)= (e^{i\rho_{r0}t}z_0, \ldots , e^{i\rho_{rN}t}z_N ) , 
\end{equation}
generated by $I_r$, where $t\in \mathbb{R}$ and $r=0,1, \ldots, N$. The resonance condition \eqref{omega} implies that the flows $\sigma_r$ are periodic
\begin{equation}\label{potok}
\sigma_r (t+T_r ) = \sigma_r (t) 
\end{equation}
for $r=1,2, \ldots, N$. We will assume that $T_1, \ldots, T_N$ are minimal periods. Expressing $\sigma_r(t) $ in the coordinates $(I_0, \ldots, I_N, \psi_0, \ldots, \psi_N )$ we find that
\begin{equation}
\sigma_r (t) (I_0, \ldots, I_N, \psi_0, \ldots, \psi_N ) = (I_0, \ldots, I_N, \psi_0, \ldots, \psi_r +t, \ldots, \psi_N ). 
\end{equation}
From the above it follows that, for $r=1,2, \ldots, N$, one can assume 
\begin{equation}\label{ter}
0< \psi_r \leq T_r. 
\end{equation}
Because of \eqref{omega} the variable $\psi_0$ depends on $\phi_0, \ldots, \phi_N$ as follows 
\begin{equation}\label{psizerodef}
\psi_0 = \sum_{j=0}^N l_j \phi_j . 
\end{equation}
In particular 
\begin{equation}\label{te0}
2\pi \sum_{i\in N_n } l_i  < \psi_0 \leq 2\pi \sum_{i\in N_p} l_i , 
\end{equation}
where $N_n := \{ 0\leq i \leq N : l_i <0 \}$ and $N_p := \{ 0\leq i \leq N : l_i > 0 \}$.

According to \eqref{int} the coordinates $(I_0, \ldots, I_N)$ belong to the cone $\Lambda^{N+1}\subset \mathbb{R}^{N+1} $ defined by inequalities 
\begin{align}
 \notag l_0I_0 + \sum_{j=1}^N \kappa_{0j}I_j > & 0 ,\\
\label{dots}
 \cdots \qquad \qquad &   \\
\notag l_NI_0 + \sum_{j=1}^N \kappa_{Nj}I_j > & 0 . 
\end{align}

 In coordinates \eqref{int}, \eqref{psil} the Poisson bracket \eqref{pb} and the Hamiltonian \eqref{H} assumes the form 
\begin{equation}\label{pb2}
\{f,g\}=  \sum_{n=0}^N\left(\frac{\partial f}{\partial I_n}\frac{\partial g}{\partial \psi_n}- \frac{\partial g}{\partial I_n}\frac{\partial f}{\partial \psi_n}\right)
\end{equation}
and 
\begin{equation}\label{H2}
H= H_0(I_0, \ldots , I_N )+2 \sqrt{\mathcal{G}_0(I_0, \ldots , I_N)}\cos \psi_0 , 
\end{equation}
respectively. The functions $H_0(I_0, \ldots , I_N)$ and $\mathcal{G}_0(I_0, \ldots , I_N )$ are defined as the superposition of functions $h_0(|z_0|^2, \ldots , |z_N|^2 )$ and \\
$|g_0(|z_0|^2, \ldots , |z_N|^2)|^2(|z_0|^{2|l_0|}\ldots |z_N|^{2|l_N|})$ with the linear map 
\begin{equation}\label{kapa}
 |z_j|^2 = \sum_{k=0}^N \kappa_{jk}I_k ,
 \end{equation}
 i.e. 
 \begin{equation}\label{funkcjagiezero}
 \mathcal{G}_0 (I_0, \ldots, I_N ) = g_0 \left( \sum_{j=0}^N \kappa_{0j} I_j, \ldots , \sum_{j=0}^N \kappa_{Nj}I_j\right)^2 \left(\sum_{j=0}^N \kappa_{0j}I_j\right)^{l_0} \ldots \left(\sum_{j=0}^N \kappa_{Nj}I_j\right)^{l_N} . 
 \end{equation}
Hence, the Hamilton equations defined by \eqref{pb2} and \eqref{H2} come out as 
\begin{align}
\label{eq-i0}
 \frac{dI_0}{dt} &=2 \sqrt{\mathcal{G}_0 (I_0, \ldots , I_N)} \sin \psi_0 \\
\label{eq-ik}
 \frac{dI_k}{dt} &=0 \\
\label{eq-psi0}
 \frac{d\psi_0 }{dt} &=\frac{\partial H_0}{\partial I_0} (I_0, \ldots , I_N ) +  \frac{\partial \mathcal{G}_0}{\partial I_0} (I_0, \ldots , I_N ) \frac{\cos \psi_0 }{\sqrt{\mathcal{G}_0(I_0, \ldots , I_N)}} \\
\label{eq-psik}
 \frac{d\psi_k }{dt} &=\frac{\partial H_0}{\partial I_k} (I_0, \ldots , I_N ) + \frac{\partial \mathcal{G}_0}{\partial I_k} (I_0, \ldots , I_N ) \frac{\cos \psi_0 }{\sqrt{\mathcal{G}_0(I_0, \ldots , I_N)}} , 
 \end{align}
 where $k=1,2,\ldots , N$. Since,  
 \begin{equation}
\{ I_k , H\} = 0, 
\end{equation}
for $k=1,\ldots , N$, we will consider the integrals of motion $I_1, \ldots , I_N$ as the components of the momentum map 
\begin{equation}\label{mmap}
\textbf{J} (I_0, \ldots , I_N, \psi_0, \ldots , \psi_N ) = \left(\begin{array}{c}
I_1 \\
\vdots \\
I_N 
\end{array}\right) ,
\end{equation}
where we identified $\mathbb{R}^N$ with the dual of Lie algebra of the $N$-dimensional torus $\mathbb{T}^N= \mathbb{S}^1\times \ldots \times \mathbb{S}^1$.

It is a consequence of \eqref{omega} that the momentum map $\textbf{J} : \Omega^{N+1} \to \mathbb{R}^N $ is a submersion. So, the level set $\textbf{J}^{-1}(c_1, \ldots , c_N )$ of $(c_1, \ldots , c_N)^T \in \textbf{J}(\Omega^{N+1})$ is a real submanifold of $\Omega^{N+1}$. In order to describe the structure of $\textbf{J}^{-1}(c_1, \ldots, c_N)$ we note that from \eqref{dots} it follows that
\begin{equation}
a< I_0 < b, 
\end{equation}
where 
\begin{align}
a:= \max_{i \in N_p } \left\{-\frac{1}{l_i}\sum_{j=1}^N \kappa_{ij}c_j \right\}, \qquad 
b:= \min_{i\in N_n} \left\{ -\frac{1}{l_i}\sum_{j=1}^N \kappa_{ij}c_j \right\} ,
\end{align}
if $(I_0, I_1, \ldots, I_N, \psi_0, \psi_1, \ldots, \psi_N ) \in \textbf{J}^{-1}(c_1, \ldots, c_N)$. Recall here that $\psi_0 $ and $\psi_1, \ldots, \psi_N $ are limited by \eqref{te0} and \eqref{ter}, respectively. So, we have $\textbf{J}^{-1}(c_1, \ldots, c_N) \cong ]a,b[ \times \mathbb{S}^1 \times \mathbb{T}^N$ and thus, $\textbf{J}^{-1}(c_1, \ldots, c_N ) /\mathbb{T}^N \cong ]a, b[ \times \mathbb{S}^1$. Summarizing, we observe that $\textbf{J}^{-1}(c_1, \ldots , c_N) \to \textbf{J}^{-1}(c_1, \ldots , c_N )/ \mathbb{T}^N $ is a trivial $\mathbb{T}^N$-pricipal bundle over the reduced symplectic manifold $\textbf{J}^{-1}(c_1, \ldots , c_N )/ \mathbb{T}^N $.

We note that $(I_0, \psi_0)$ are canonical coordinates on $\textbf{J}^{-1}(c_1, \ldots , c_N )/ \mathbb{T}^N $, i.e. the reduced Poisson bracket of $F, G \in C^\infty (\textbf{J}^{-1}(c_1, \ldots , c_N )/ \mathbb{T}^N )$ is given by 
\begin{equation}\label{pb3}
\{F, G\} = \frac{\partial F}{\partial I_0}\frac{\partial G}{\partial \psi_0} - \frac{\partial G}{\partial I_0}\frac{\partial F}{\partial \psi_0} .
\end{equation}
Therefore, the Hamilton equations (\ref{eq-i0}-\ref{eq-psik}) reduce to 
\begin{align}
\label{eq-i02}
 \frac{dI_0}{dt} &=2 \sqrt{\mathcal{G}_0 (I_0, c_1\ldots , c_N)} \sin \psi_0 , \\
\label{eq-psi02}
 \frac{d\psi_0 }{dt} &=\frac{\partial H_0}{\partial I_0} (I_0, c_1\ldots , c_N ) +  \frac{\partial \mathcal{G}_0}{\partial I_0} (I_0,c_1 \ldots , c_N ) \frac{\cos \psi_0 }{\sqrt{\mathcal{G}_0(I_0, c_1,  \ldots , c_N)}} ,
 \end{align}
and one can integrate them by quadratures using  
\begin{equation}\label{energy-law}
H_0(I_0, c_1, \ldots , c_N )+2 \sqrt{\mathcal{G}_0(I_0, c_1,  \ldots , c_N)}\cos \psi_0 = E= const . 
\end{equation}
Namely, from \eqref{eq-i02} and \eqref{energy-law} one obtains the differential equation 
\begin{equation}\label{difi0}
\left(\frac{dI_0}{dt} (t)\right)^2 = 4 \mathcal{G}_0 (I_0 (t), c_1, \ldots , c_N ) - \left(E- H_0(I_0 (t), c_1, \ldots , c_N ) \right)^2, 
\end{equation}
on $I_0(t)$, which could be solved by the separation of variables. Next, substituting $I_0(t)$ into \eqref{eq-psi02} we find $\psi_0 (t)$. Having obtained $I_0(t)$ and $\psi_0(t)$ we find $\psi_k(t)$ integrating both sides of \eqref{eq-psik}.

In order to visualize the geometry of the reduced symplectic manifold $\textbf{J}^{-1}(c_1, \ldots , c_N )/ \mathbb{T}^N $ let us introduce a map $\mathcal{F}:\Omega^{N+1} \to \mathbb{C}$ given by
\begin{multline}\label{cfz}
z=x+iy =\mathcal{F}(z_0, \ldots , z_N) := g_0(|z_0|^2,|z_1|^2,\ldots, |z_N|^2) z_0^{l_0}\cdots z_N^{l_N}=\\
 = \sqrt{\mathcal{G}_0(I_0, \ldots , I_N)} e^{i\psi_0} ,
 \end{multline}
which is constant on the orbits of $\mathbb{T}^N$ and thus, can be considered as a function of arguments $I_0, \ldots , I_N, \psi_0$. We observe that the variables $I_0, I_1, \ldots, I_N, x$ and $y$ are functionally closed with respect to the Poisson bracket \eqref{pb}, i.e. one has
\begin{equation}\label{pbx0x}
\{I_0,x\}=- y,
\end{equation}

\begin{equation}\label{pbx0y}
 \{I_0,y\}= x,
 \end{equation}
 
\begin{equation}\label{pbxy} 
\{x,y\}= \frac{1}{2} \frac{\partial \mathcal{G}_0}{\partial I_0}(I_0,I_1,\ldots I_N),
\end{equation}
 \begin{equation}\label{2.18}
 \{I_k,x\}=\{I_k,y\}=0,
 \end{equation}
 for $k,l=1,2, \ldots , N$.

Therefore, the variables $I_0, \ldots, I_N, x,y$ generate a Poisson subalgebra $\mathcal{K}_{\mathcal{G}_0} (\Omega^{N+1}) $ of the standard Poisson algebra $(C^\infty (\Omega^{N+1} ), \{ \cdot, \cdot \})$ of all smooth functions on $\Omega^{N+1} $. Let us note here that $\mathcal{K}_{\mathcal{G}_0}(\Omega^{N+1})$ is determined by the function $\mathcal{G}_0\in C^\infty (\Omega^{N+1})$ which depends on the choice of $g_0 \in C^\infty (\mathbb{R}^{N+1})$ and $(l_0, \ldots , l_N)^T \in \mathbb{Z}^{N+1}$. The Hamiltonian \eqref{H} belongs to $\mathcal{K}_{\mathcal{G}_0} (\Omega^{N+1}) $ and integrals of motion $I_1, \ldots , I_N$ generate the center of $\mathcal{K}_{\mathcal{G}_0} (\Omega^{N+1}) $. 

Since functions $x,y, I_0 \in C^\infty (\Omega^{N+1})$ are invariants of $\mathbb{T}^N $, they define the corresponding functions on the reduced phase space $\textbf{J}^{-1}(c_1, \ldots , c_N) / \mathbb{T}^N $. Hence, there is a map 
\begin{equation}
\Phi_{c_1, \ldots, c_N} (I_0, \psi_0) := \left(\begin{array}{c} 
\sqrt{\mathcal{G}_0(I_0, c_1, \ldots, c_N)} \cos \psi_0 \\
\sqrt{\mathcal{G}_0(I_0, c_1, \ldots, c_N)} \sin \psi_0\\
I_0 
\end{array}\right) 
\end{equation}
 of $]a,b[ \times \mathbb{S}^1$ onto the circularly symmetric surface $\mathcal{C}^{-1}(0)$ in $\mathbb{R}^2\times ]a,b[$ given by the equation
\begin{equation}\label{casimir}
\mathcal{C}(x,y, I_0):=-\frac12 (x^2+y^2 - \mathcal{G}_0(I_0,c_1,\ldots, c_N ) ) =0
\end{equation}
 on $(x,y,I_0)^T \in \mathbb{R}^2\times ]a,b[$.  It follows from \eqref{te0} that $\Phi_{c_1, \ldots, c_N}$ is a $\sum_{i=0}^N |l_i|$-fold covering of $\mathcal{C}^{-1}(0)$. Following \cite{holm} we will call $\mathcal{C}^{-1}(0) $ the \textbf{Kummer shape}. The special case $H_0=0$, $g_0=const$ and  $(l_0, l_1, l_2)=(1,-1,1)$ was investigated in \cite{Luther}, \cite{Alber}, \cite{Alber1}, where $\mathcal{C}^{-1}(0)$ was called a three-wave surface.

Taking the Poisson algebra $(C^\infty (\mathbb{R}^3), \{\cdot , \cdot \}_\mathcal{C})$ with the Nambu bracket 
 \begin{equation}\label{nambu}
\{f,g\}_\mathcal{C}:=\det [\nabla \mathcal{C},\nabla f, \nabla g]
\end{equation}
as the Poisson bracket, where $f,g \in C^\infty (\mathbb{R}^3)$ and $ \nabla f = \left( \frac{\partial f}{\partial x} , \frac{\partial f}{\partial y}, \frac{ \partial f}{\partial I_0 } \right)^T$, we find that $\Phi_{c_1, \ldots, c_N}:\textbf{J}^{-1}(c_1, \ldots , c_N) / \mathbb{T}^N \to \mathbb{R}^3$ is a Poisson map. This fact is a direct consequence of the relations (\ref{pbx0x}-\ref{pbxy}) and the reduction procedure. Note that $\mathcal{C}$ is a Casimir for the Poisson bracket $\{\cdot, \cdot \}_{\mathcal{C}}$. So, the Kummer shape $\mathcal{C}^{-1}(0)$  is a symplectic leaf and $\Phi_{c_1, \ldots, c_N}:\textbf{J}^{-1}(c_1, \ldots , c_N) / \mathbb{T}^N \to \mathcal{C}^{-1}(0)$ is a symplectic map.  We recall that $\mathcal{G}_0(I_0, c_1, \ldots , c_N) >0$, which implies that $\mathcal{C}^{-1}(0)$  does not have singular points, i.e. it is a submanifold of $\mathbb{R}^3$

The functions $x,y, I_0 \in \mathcal{K}_{\mathcal{G}_0}(\Omega^{N+1})$ after reduction to $\textbf{J}^{-1}(c_1, \ldots , c_N)/ \mathbb{T}^N$ satisfy \eqref{pbx0x}, \eqref{pbx0y}) as well as 
\begin{equation}
\{x, y \} = \frac{1}{2} \frac{\partial \mathcal{G}_0}{\partial I_0}(I_0, c_1, \ldots , c_N) .
\end{equation}
Thus and from the definition of $\{\cdot , \cdot \}_\mathcal{C}$ 
it follows that they generate a Poisson algebra $\mathcal{K}_{\mathcal{G}_0}(c_1, \ldots , c_N) $ isomorphic to $(C^\infty (\mathbb{R}^3), \{\cdot , \cdot \}_{\mathcal{C}}) $. This Poisson algebra is the reduction of the Poisson subalgebra $\mathcal{K}_{\mathcal{G}_0}(\Omega^{N+1}) \subset C^\infty (\Omega^{N+1})$. We shall call $\mathcal{K}_{\mathcal{G}_0}(c_1, \ldots , c_N)$ the \textbf{classical Kummer shape algebra}.

The Hamiltonian \eqref{H2} in coordinates $(x, y, I_0)$ has the form 
\begin{equation}\label{H3}
H= H_0(I_0, c_1, \ldots, c_N) +2x 
\end{equation}
and the corresponding 
Hamilton equations for $x,y$ and $I_0$ are
\begin{align}\label{i_0t}
\frac{dI_0}{dt} &=2y \\
\frac{dx}{dt} &=-y \frac{\partial H_0}{\partial I_0} (I_0, c_1, \ldots , c_N) \\
\label{dyt} \frac{dy}{dt} &=x \frac{\partial H_0}{\partial I_0} (I_0, c_1, \ldots, c_N) + \frac{\partial \mathcal{G}_0}{\partial I_0} (I_0, c_1, \ldots , c_N ). 
\end{align}
They lead to the dependences: 
\begin{equation}\label{xt}
x(t) = \frac{1}{2} \left(E- H_0(I_0(t), c_1, \ldots , c_N)\right)
\end{equation}
\begin{equation}\label{yt}
y(t) = \frac{1}{2} \frac{dI_0}{dt} (t) ,
\end{equation}
where $I_0(t)$ is the solution of equation \eqref{difi0}.

Let us note also that the trajectory $ \mathbb{R}\ni t\mapsto (x(t), y(t), I_0(t))^T \in \mathcal{C}^{-1}(0)$, defined by (\ref{i_0t}-\ref{dyt}), can also be obtained as the intersection $\mathcal{C}^{-1}(0) \cap H^{-1}(E) $ of the Kummer shape $\mathcal{C}^{-1}(0)$ with the level set $H^{-1}(E)$ of the Hamiltonian \eqref{H3}. 

\section{Quantum Kummer shape algebra}\label{sec:qkummer}

Following \cite{AO},  we will repeat the considerations of the previous section in the case of a quantum analogue of the classical hamiltonian system presented there. As the quantum counterpart of the Hamiltonian \eqref{H} we take

\begin{multline}\label{qH}
\textbf{H}= h_0(a_0^*a_0, ... , a_N^*a_N) + g_0(a_0^*a_0, ..., a_N^*a_N ) a_0^{l_0}... a_N^{l_N} +\\
+ a_0^{-l_0}...a_N^{-l_N}g_0(a_0^*a_0, ..., a_N^*a_N), 
\end{multline}
where, similarly to \eqref{zet}, we have employed the convention 
\begin{equation}
 a^{l_i}_i=\left\{\begin{array}{ll}

a^{l_i}_i & \textrm{if }l_i\geq 0  \\ 

 (a_i^*)^{-l_i} & \textrm{if }l_i<0

                      \end{array}\right. .
\end{equation}

Everywhere below we will keep Planck's constant $\hbar $ in the canonical commutation relations

\begin{equation}\label{comm}
[a_i, a_j^*] = \hbar \delta_{ij}, \qquad [a_i, a_j] = 0, \qquad [a_i^*, a_j^*] = 0, 
\end{equation}
for the annihilation and creation operators.

Following the classical case (see \eqref{int} and \eqref{cfz}), we introduce the operators 

\begin{equation}\label{ao}
A:= g_0(a_0^*a_0,..., a_N^*a_N) a_0^{l_0}...a_N^{l_N}, 
\end{equation}

\begin{equation}\label{ai}
A_i:= \sum_{j=0}^N \rho_{ij}a_j^*a_j,  
\end{equation}
where $ i=0,1,...,N$, and  $\rho = (\rho_{ij} )$ is the matrix defined in \eqref{omega}.
The operators $A_0, A_1, ... , A_N, A, A^*$ satisfy the commutation relations
\begin{equation}\label{comrel}
[A_0, A] = -\hbar A, \qquad [A_0, A^*] = \hbar A^*, 
\end{equation}
\begin{equation}
[A, A_i]= [A^*, A_i] = 0, 
\end{equation}
where $i=1,2,..., N$, and 
\begin{equation}\label{comij}
[A_i, A_j] = 0 , 
\end{equation}
where $i,j=0,1,..., N$. 
From \eqref{comm} and \eqref{ao} one obtains
\begin{equation}\label{8a}
A^*A = |g_0(a_0^*a_0-\hbar l_0,..., a_N^*a_N-\hbar l_N)|^2\mathcal{P}_{l_0}(a_0^*a_0-\hbar l_0 )... \mathcal{P}_{l_N}(a_N^*a_N- \hbar l_N ), 
\end{equation}
\begin{equation}\label{a8}
AA^* = |g_0(a_0^*a_0,..., a_N^*a_N)|^2\mathcal{P}_{l_0}(a_0^*a_0)... \mathcal{P}_{l_N}(a_N^*a_N ), 
\end{equation}
where the polynomials $\mathcal{P}_{l_i}$ are defined by 
\begin{equation}
\mathcal{P}_{l_i}(a_i^*a_i):=\left\{\begin{array}{ll}

a^{l_i}_i(a_i^*)^{l_i} =(a_i^*a_i+\hbar )...(a_i^*a_i +l_i\hbar )& \textrm{if }l_i > 0  \\ 

1 & \textrm{if }l_i=0 \\

 (a_i^*)^{-l_i}a_i^{-l_i}=a_i^*a_i(a_i^*a_i-\hbar )... (a_i^*a_i -(-l_i-1)\hbar ) & \textrm{if }l_i<0

                      \end{array}\right.
\end{equation}

Replacing in \eqref{a8} and \eqref{8a} the occupation number operators \\
$a_0^*a_0,... , a_N^*a_N$ by the operators $A_0,A_1, ... , A_N$  one finds 
\begin{equation}\label{aa}
AA^*= \mathcal{G}_\hbar (A_0, A_1,..., A_N), 
\end{equation}
\begin{equation}\label{aaa}
A^*A= \mathcal{G}_\hbar (A_0-\hbar, A_1,..., A_N), 
\end{equation}
where the function $\mathcal{G}_\hbar $ is defined by 
\begin{multline}\label{ghf}
\mathcal{G}_\hbar (A_0, \ldots , A_N ) := \\
g_0\left(\sum_{j=0}^N \kappa_{0j}A_j,..., \sum_{j=0}^N \kappa_{Nj}A_j\right)^2 \mathcal{P}_{l_0}\left(\sum_{j=0}^N \kappa_{0j}A_j\right)... \mathcal{P}_{l_N}\left(\sum_{j=0}^N \kappa_{Nj}A_j \right).
\end{multline}
In terms of the operators $A_0, A_1, ..., A_N, A, A^*$ the Hamiltonian \eqref{qH} is written as follows
\begin{equation}\label{rH}
\textbf{H}= H_0(A_0, A_1, \ldots , A_N ) + A+ A^*, 
\end{equation}
where the function $H_0$ is defined as the superposition of the function $h_0$ with the linear map inverse to \eqref{ai}.

One can easily see that $[A_i, \textbf{H}]= 0$ for $i=1,2,\ldots , N$. So we have commuting integrals of motion: $A_1, \ldots , A_N, $ which also commute with $A_0$. Notice here that the operators $A_0,A_1, \ldots , A_N$ are diagonalized in the standard Fock basis 
\begin{equation}\label{Fockbasis}
\left| n_0, n_1, \ldots , n_N \right> := \frac{1}{\sqrt{n_0!\ldots  n_N!}}\hbar^{-\frac{1}{2}(n_0+...+n_N)}(a_0^*)^{n_0}\ldots (a_N^*)^{n_N} |0,\ldots , 0\rangle ,  
\end{equation}
where $n_i \in \mathbb{Z}_+\cup \{ 0\}$, with the eigenvalues $c_0, c_1, \ldots , c_N $ related to $n_0, n_1, \ldots, n_N$ by
\begin{equation}\label{lambdai}
 c_i= \hbar \sum_{j=0}^N \rho_{ij} n_j , \qquad i=0,1,..., N. 
\end{equation}
We will use them for a new parametrization $\{|c_0, c_1 , \ldots , c_N \rangle \}$ of the Fock basis $\{|n_0, n_1,..., n_N \rangle \}$.

Taking into account the above, we can reduce the quantum system described by the Hamiltonian \eqref{qH} to the Hilbert subspace $\mathcal{H}_{c_1 , \ldots , c_N }\subset \mathcal{H} $ spanned by the eigenvectors $|c_0 , c_1 , \ldots , c_N \rangle $ of $A_0$, 
with fixed eigenvalues $ c_1,  \ldots ,  c_N$ of the operators $A_1, \ldots , A_N$. Let us describe the dependence of $|c_0 , c_1 , \ldots , c_N \rangle \in \mathcal{H}_{c_1 , \ldots , c_N }$ on the choice of the exponents $l_0, l_1, \ldots , l_N \in \mathbb{Z}$. First we recall that the vectors $|n_0, n_1, \ldots , n_N \rangle $, where $n_0, n_1, \ldots , n_N \in \mathbb{N}\cup \{0\}$, form the set of all common eigenvectors for the operators $A_0, A_1, \ldots , A_N$. It follows from \eqref{omega} and \eqref{lambdai} that if $ |n_0, n_1, \ldots , n_N \rangle ,  | v_0, v_1, \ldots , v_N \rangle , \in \mathcal{H}_{c_1 , \ldots , c_N }$, then there exists $n\in \mathbb{Z}$ such that 
\begin{equation}\label{enki}
\left(\begin{array}{c}
n_0 \\
n_1 \\
\vdots \\
n_N 
\end{array}\right) = 
\left(\begin{array}{c}
v_0 \\
v_1 \\
\vdots \\
v_N 
\end{array}\right) + n
\left(\begin{array}{c}
l_0 \\
l_1 \\
\vdots \\
l_N 
\end{array}\right) .
\end{equation}
From \eqref{omega} and the equation \eqref{lambdai} and \eqref{enki} we find that the eigenvalues $ \tilde{c}_0, \tilde{c}_1 , \ldots , \tilde{c}_N $ and $ c_0, c_1, \ldots , c_N $ of $A_0, A_1, \ldots, A_N $ corresponding to $|n_0, \ldots, n_N \rangle $ and $|v_0, \ldots, v_N \rangle $, respectively, are related by 
\begin{equation}
\left(\begin{array}{c}
\tilde{c}_0 \\
\tilde{c}_1 \\
\vdots \\
\tilde{c}_N 
\end{array}\right) = 
\left(\begin{array}{c}
c_0 +\hbar n \\
c_1 \\
\vdots \\
c_N 
\end{array}\right) ,
\end{equation}
where the parameters $c_1, \ldots, c_N$ are fixed after reduction to $\mathcal{H}_{c_1, \ldots, c_N}$.

In the following we assume that $c_0 $ satisfies 
\begin{equation}\label{v1}
\mathcal{G}_\hbar (c_0-\hbar,  c_1,\ldots ,  c_N) = 0, 
\end{equation}
which is equivalent to the assumption that $|c_0, c_1, \ldots , c_N \rangle $ is a vacuum state of the annihilation operator $\textbf{A}$, i.e. one has 
\begin{equation}\label{vacuum}
 \textbf{A}|c_0, c_1, \ldots , c_N \rangle =0 .
\end{equation}
The equivalence of the conditions \eqref{v1} and \eqref{vacuum} can be easily seen from the relation \eqref{aaa}. Let us note that 
the basis of Hilbert subspace $\mathcal{H}_{c_1 , \ldots ,  c_N }$ is generated from $|c_0, c_1, \ldots , c_N \rangle $ by the creation operator $\textbf{A}^*$.

The condition \eqref{vacuum} forces $n$ to be non-negative. The upper bound on $n$ depends on the choice of $(l_0, l_1, \ldots , l_N)^T \in \mathbb{Z}^{N+1}$ in the following way. If $l_k \geq 0$ for all $k\in \{0,1, \ldots , N\}$, then $n$ runs over all $\mathbb{N}\cup \{0\}$ and thus, $\dim \mathcal{H}_{c_1 , \ldots ,  c_N } = \infty $. In the opposite case one has $0 \leq n \leq L$, where $L\in \mathbb{N}$, so $\mathcal{H}_{c_1 , \ldots ,  c_N }$ is isomorphic to $\mathbb{C}^{L+1}$. We mention here that $L$ depends on the choice of the vacuum vector $|c_0, c_1, \ldots , c_N \rangle$. We refer to \cite{AO}, where more detailed discussion of the above facts can be found.

It follows from the above that the operators $A_0, A, A^*$ after reduction to $\mathcal{H}_{c_1 , \ldots ,  c_N }$ are given by
\begin{equation}\label{ra0}
\textbf{A}_0|c_0 +\hbar n, c_1, \ldots , c_N \rangle = (c_0+\hbar n) |c_0 + \hbar n, c_1, \ldots , c_N \rangle
\end{equation}
\begin{equation}
\textbf{A}|c_0 +\hbar n, c_1, \ldots , c_N \rangle = 
\sqrt{\mathcal{G}_\hbar (c_0+\hbar (n-1),  c_1, \ldots ,  c_N)}|c_0 +\hbar (n-1), c_1, \ldots , c_N \rangle
\end{equation}
\begin{equation}\label{rastar}
\textbf{A}^*|c_0 +\hbar n, c_1, \ldots , c_N \rangle = \sqrt{\mathcal{G}_\hbar ( c_0+\hbar n,  c_1, \ldots ,  c_N)}|c_0 +\hbar (n+1), c_1, \ldots , c_N \rangle.
\end{equation}

We denote by $\mathcal{Q}_{\mathcal{G}_\hbar } (\mathcal{H}_{c_1 , \ldots ,  c_N } )$ the operator algebra generated by the reduced operators $\textbf{A}, \textbf{A}^*$ and $\textbf{A}_0$. In accordance with the classical case, we will call this algebra a \textbf{quantum Kummer shape algebra}.


After reduction to $\mathcal{H}_{c_1, \ldots , c_N }$, the relations \eqref{comrel} and \eqref{aa}, \eqref{aaa} for $\textbf{A}_0$, the annihilation $\textbf{A}$ and the creation $\textbf{A}^*$ operators are given now by 
\begin{equation}\label{rel1}
[\textbf{A}_0, \textbf{A}] = -\hbar \textbf{A}, \qquad [\textbf{A}_0, \textbf{A}^*] = \hbar \textbf{A}^*, 
\end{equation}
\begin{equation}\label{rel3}
\textbf{A}^*\textbf{A} = \mathcal{G}_\hbar (\textbf{A}_0 -\hbar,  c_1, \ldots ,  c_N )
\end{equation}
\begin{equation}\label{rel2}
\textbf{A}\textbf{A}^* = \mathcal{G}_\hbar (\textbf{A}_0 , c_1 , \ldots ,  c_N ). 
\end{equation}

After reduction to $\mathcal{H}_{c_1 , \ldots ,  c_N }$ the Hamiltonian \eqref{qH}  is expressed by $\textbf{A}_0, \textbf{A} $ and $\textbf{A}^*$ as follows
\begin{equation}\label{Hr}
\textbf{H}= H_0(\textbf{A}_0, c_1, \ldots ,  c_N ) +\textbf{A}+ \textbf{A}^*.
\end{equation}

As we have shown, two complementary subcases are possible :\\
\textbf{(i)} the vector $(l_0, \ldots , l_N )^T \in \mathbb{Z}^{N+1}$ has at least two components of the different sign; \\
\textbf{(ii)} all nonzero components of $(l_0, \ldots , l_N )^T $ are positive. \\

In the case \textbf{(i)} one has the splitting 
\begin{equation}\label{split}
\mathcal{H} = \bigoplus_{c_1, \ldots , c_N } \mathcal{H}_{c_1, \ldots , c_N }
\end{equation}
of the Hilbert space $\mathcal{H}$ into finite dimensional Hilbert subspaces which are invariant with respect to the Hamiltonian \eqref{Hr}, i.e. $\textbf{H}\mathcal{H}_{c_1, \ldots , c_N } \subset \mathcal{H}_{c_1, \ldots , c_N }$. Hence, one can reduce the spectral decomposition of the Hamiltonian $\textbf{H}$ to the spectral decompositions of three-diagonal hermitian matrices $\textbf{H}\mid_{\mathcal{H}_{c_1, \ldots , c_N }} : \mathbb{C}^{L+1} \to \mathbb{C}^{L+1}$. The Hamiltonians of this type describe various non-linear phenomena in quantum optics, e.g see \cite{AO2}, \cite{ATer}. See also Example \ref{krawtchouk}, which is presented in Section \ref{sec:examples}. 

Finally, let us note that in this case the quantum Kummer shape $\mathcal{Q}_{\mathcal{G}_\hbar } (\mathcal{H}_{c_1, \ldots , c_N } )$ is isomorphic to the algebra of $(L+1)\times (L+1)$ complex matrices.

In the case \textbf{(ii)} all components of the splitting \eqref{split} are infinite dimensional Hilbert subspaces.
 Therefore, Hamiltonian \eqref{qH} splits now into the Jacobi operators, see \cite{Tesch}. Hence, in some special cases, it can be integrated using the theory of orthogonal polynomials, e.g. see Section V in \cite{AO2}. 

The operator algebra $\mathcal{Q}_{\mathcal{G}_\hbar } (\mathcal{H}_{c_1, \ldots , c_N })$ is defined by the choice of the coupling function $g_0$ and vector $(l_0, l_1, \ldots , l_N )^T \in \mathbb{Z}^{N+1}$, which determine the structural function $\mathcal{G}_\hbar $. Having fixed  the structural function $\mathcal{G}_\hbar $ and Hilbert subspace $\mathcal{H}_{c_1, \ldots , c_N}$, one can forget about the previous steps of our construction and consider $\mathcal{Q}_{\mathcal{G}_\hbar } (\mathcal{H}_{c_1, \ldots , c_N })$ as an object depending only on the choice of $\mathcal{G}_\hbar $.

As we see from \eqref{ra0}, if $\dim_{\mathbb{C}}\mathcal{H}_{c_1, \ldots, c_N}= \infty$ then $\textbf{A}_0$ is an unbounded operator. However, depending on the choice of $\mathcal{G}_\hbar $, the annihilation and creation operators $\textbf{A}$ and $\textbf{A}^*$ can both be bounded. For example, one can assume that 
\begin{equation}
\mathcal{G}_\hbar (\textbf{A}_0,  c_1, \ldots ,  c_N ) = \mathcal{R} (q^{\frac{1}{\alpha} (\textbf{A}_0 - c_0 )}), 
\end{equation}
where $0< q <1 $ and $\alpha $ is a constant which has dimension of an action. If $\mathcal{R} $ is bounded on the interval $q^{-\frac{1}{\alpha}c_0} [0, 1]$, then $\textbf{A}$ and $\textbf{A}^*$ are bounded operators. Taking the bounded operator 
\begin{equation}\label{qoperator}
\textbf{Q}:= q^{\frac{1}{\alpha} (\textbf{A}_0 - c_0) } , 
\end{equation}
instead of $\textbf{A}_0$ we rewrite the relations (\ref{rel1}-\ref{rel2}) as follows 
\begin{equation}\label{com1}
\textbf{A}\textbf{Q}=q^{\frac{\hbar}{\alpha }}\textbf{Q}\textbf{A}, \quad \textbf{Q}\textbf{A}^*= q^{\frac{\hbar}{\alpha}}\textbf{A}^*\textbf{Q},
\end{equation}
\begin{equation}\label{com2}
\textbf{A}^*\textbf{A} = \mathcal{R}(q^{-\frac{\hbar}{\alpha}}\textbf{Q}), \quad \textbf{A}\textbf{A}^* = \mathcal{R}(\textbf{Q}).
\end{equation}

The operator $C^*$-algebras defined by the above relations were investigated in \cite{AO1}. They are strictly related to the theory of the basic hypergeometric series and have many interesting applications. We will illustrate this in Section \ref{sec:examples} presenting a suitable example. Taking in \eqref{com2} the quadratic polynomials as the structural function $\mathcal{R}$  we obtain various Podle\'s hemispheres \cite{Podles1}.

\section{Coherent states, star product and reduction}\label{sec:4}

In this section we will discuss the correspondence between classical and quantum Kummer algebras. Our consideration will be based on the notion of a coherent state map, see \cite{cohS}, i.e. a symplectic map $\mathcal{K}: M \to \mathbb{C}\mathbb{P}(\mathcal{H})$ of the classical phase space $M$ (a symplectic manifold) into the quantum phase space $\mathbb{C}\mathbb{P}(\mathcal{H})$, as well as on the related notion of a covariant symbol of an operator. In particular, we will show that this correspondence is consistent with the procedures of classical and quantum reduction.

In the beginning we recall the definition of the Glauber coherent states for a system of $N+1$ non-interacting modes (harmonic oscillators): 
\begin{equation}\label{Glauberstate}
|z_0, \ldots , z_N \rangle := \sum_{n_0,..., n_N=0}^\infty \frac{z_0^{n_0} \cdots z_N^{n_N}}{\sqrt{n_0! \ldots n_N! }} \hbar^{-\frac{1}{2}(n_0+...+n_N)} |n_0, \ldots , n_N \rangle ,
\end{equation}
where $z_0, \ldots , z_N \in \mathbb{C}$ and $|n_0, \ldots , n_N \rangle $ are the elements of the Fock basis of the Hilbert space $\mathcal{H} $, defined in \eqref{Fockbasis}.

The covariant symbol $\langle F \rangle : \mathbb{C}^{N+1} \to \mathbb{C} $ of an operator 
\begin{equation}
F= \sum_{m_0,..., m_N,n_0, ..., n_N=0}^\infty f_{\bar m_0, ..., \bar m_N , n_0, ..., n_N } (a_0^*)^{m_0}... (a_N^*)^{m_N} a_0^{n_0} ... a_N^{n_N}
\end{equation}
is defined by the mean value 
\begin{multline}\label{covf}
\langle F \rangle (\bar z_0, ... , \bar z_N , z_0, ..., z_N ) := \frac{\langle z_0, ..., z_N |F| z_0, ... , z_N \rangle }{\langle z_0, ..., z_N | z_0, ... , z_N \rangle } \\
= \sum_{m_0,..., m_N, n_0, ..., n_N=0}^\infty f_{\bar m_0, ..., \bar m_N , n_0, ..., n_N } (\bar z_0)^{m_0}... (\bar z_N)^{m_N} z_0^{n_0} ... z_N^{n_N} \\
\end{multline}
of $F$ on the coherent states.
We have assumed in \eqref{covf} that the coherent states \eqref{Glauberstate} belong to the domain of $F$.  For brevity we will write $f$ instead of $\langle F \rangle$ in the sequel.

The $*_\hbar $-product of covariant symbols $f,g \in C^\infty (\mathbb{C}^{N+1})$ of the operators $F$ and $G$ is defined in the following way: 
\begin{equation}\label{hpf}
(f*_\hbar g ) (\bar z_0 , ... , \bar z_N , z_0, ... , z_N ) := \langle FG \rangle (\bar z_0 , ... , \bar z_N , z_0, ... , z_N ). 
\end{equation}
Using the resolution 
\begin{equation}\label{glauber1}
\int_{\mathbb{C}^{N+1}} \frac{ |w_0, ... , w_N \rangle \langle w_0, ... , w_N | } {\langle w_0, ... , w_N | w_0, ... , w_N \rangle } d\nu_\hbar (\bar w_0,..., \bar w_N, w_0,..., w_N )  = \mathbbm{1},
\end{equation}
 of identity $\mathbbm{1}$, where $d\nu_\hbar (\bar w_0,..., \bar w_N, w_0,..., w_N ) = \frac{1}{(\pi \hbar )^{N+1}} d(\Re w_0 )... d(\Re w_N ) d(\Im w_0)... d(\Im w_N )$ is the Liouville measure of the symplectic manifold $\mathbb{C}^{N+1}$ normalized by a factor, one obtains from \eqref{hpf} the standard formula for $*_\hbar$-product
\begin{multline}
(f*_\hbar g ) (\bar z_0 , ... , \bar z_N , z_0, ... , z_N ) =  \\
 =\int_{\mathbb{C}^{N+1}} f(\bar z_0, \ldots , \bar z_N , w_0, \ldots , w_N ) g( \bar w_0 , \ldots , \bar w_N, z_0, \ldots , z_N )\times\\
 \times e^{-\frac{1}{\hbar } (|z_0- w_0 |^2 + \ldots + |z_N - w_N |^2 )} d\nu_\hbar (\bar w_0,..., \bar w_N, w_0,..., w_N ) =\\
 = \sum_{j_0, \ldots , j_N =0}^\infty \frac{\hbar^{j_0 + \ldots + j_N }}{j_0! \ldots j_N!} \left(\frac{\partial^{j_0}}{\partial z_0^{j_0}} \ldots \frac{\partial^{j_N}}{\partial z_N^{j_N}}\right) f(\bar z_0, \ldots , \bar z_N , z_0, \ldots , z_N ) \times \\
 \times \left(\frac{\partial^{j_0}}{\partial \bar z_0^{j_0}} \ldots \frac{\partial^{j_N}}{\partial \bar z_N^{j_N}}\right) g(\bar z_0, \ldots , \bar z_N , z_0, \ldots , z_N ). 
\end{multline}

Note here that 
\begin{equation}
f*_\hbar g \underset{\hbar \to 0 }{\longrightarrow } f \cdot g 
\end{equation}
 and 
\begin{equation}
\lim_{\hbar \to 0 } \frac{-i}{\hbar} (f*_\hbar g- g*_\hbar f )  = \{ f, g \}, 
\end{equation}
i.e. in the limit $\hbar \to 0 $ we come back to the Poisson algebra of real smooth functions on $\mathbb{C}^{N+1}$.

Using the above notions, after simple calculations, one obtains the following correspondences 
\begin{equation}\label{covak}
\langle A_k \rangle = I_k,
\end{equation}
\begin{equation}\label{cova}
\langle A \rangle \underset{\hbar \to 0 }{\longrightarrow } z,
\end{equation}
\begin{equation}
\langle \textbf{H} \rangle \underset{\hbar \to 0 }{\longrightarrow } H,
\end{equation}
between the suitable quantum and classical quantities defined in  \eqref{qH}, \eqref{ao}, \eqref{ai} and \eqref{H}, \eqref{int}, \eqref{cfz}, respectively. Note that when $g_0$ is constant then one has the equality $ \langle A \rangle = z$. See \eqref{cfz} for the definition of the variable $z$.

It follows from \eqref{covak} and \eqref{cova} that after taking the covariant symbols  \eqref{covf} for the commutation relations (\ref{comrel}-\ref{comij}) and the relation 
\begin{equation}
\frac{-i}{\hbar} [A, A^*] = \frac{-i}{\hbar} (\mathcal{G}_\hbar (A_0, A_1, ... , A_N) - \mathcal{G}_\hbar (A_0-\hbar, A_1, ... , A_N) )
\end{equation}
in the limit $\hbar \to 0$ one obtains the relations (\ref{pbx0x}-\ref{2.18}) for the corresponding classical quantities which define classical Kummer shape algebra. Summing up, we state that the system of $N+1$ quantum harmonic oscillators, considered in Section \ref{sec:qkummer}, converges in the classical limit $\hbar \to 0 $ to its classical counterpart presented in Section \ref{sec:kummer}.

Now we will apply the classical and quantum reduction procedures discussed in Section \ref{sec:kummer} and Section \ref{sec:qkummer} to the construction of the reduced coherent state map 
\begin{equation}\label{cohstat}
\mathcal{K}_{c_1, \ldots, c_N}: \mathbf{J}^{-1}(c_1, \ldots, c_N)/\mathbb{T}^N \to \mathbb{C}\mathbb{P}(\mathcal{H}_{c_1, \ldots, c_N}).
\end{equation} 
Note that the Glauber coherent state map  $K_G:\Omega^{N+1} \to \mathcal{H}$, defined in \eqref{Glauberstate}, has equivariance property 
\begin{equation}\label{potokna}
|\sigma_r (t)(z_0, \ldots ,  z_N) \rangle = e^{\frac{i}{\hbar} A_r(t)}|z_0, \ldots, z_N \rangle, 
\end{equation}
where $r=1,\ldots, N$ and $t\in \mathbb{R}$, with respect to the classical and quantum actions of $\mathbb{T}^N$. We also recall that $I_0, I_1, \ldots, I_N$ are invariants of the hamiltonian flows \eqref{flow}. Hence, passing in \eqref{potokna} from the complex canonical coordinates $(z_0, \ldots, z_N)$ to the real canonical coordinates $(I_0, I_1, \ldots, I_N, \psi_0, \psi_1, \ldots, \psi_N )$ we find that, for $r=1,\ldots, N$, one has 
\begin{multline}
P_{c_1, \ldots, c_N} |I_0, c_1, \ldots, c_N , \psi_0, \ldots, \psi_r +t, \ldots, \psi_N \rangle = \\
= e^{\frac{i}{\hbar} c_r t} P_{c_1, \ldots, c_N} |I_0, c_1, \ldots, c_N, \psi_0,  \ldots, \psi_r, \ldots, \psi_N \rangle ,
\end{multline}
if $(z_0, \ldots, z_N)^T \in \mathbf{J}^{-1}(c_1, \ldots, c_N)$, where $P_{c_1, \ldots, c_N} : \mathcal{H} \to \mathcal{H}_{c_1, \ldots, c_N}$ is the orthogonal projection of $\mathcal{H}$ on $\mathcal{H}_{c_1, \ldots, c_N}$.

 Let us notice 
 that according to \eqref{enki} the Hilbert subspace $\mathcal{H}_{c_1, \ldots, c_N} $ is spanned by the vectors $\{ |v_0+ nl_0, \ldots, v_N+nl_N \rangle \}_{n= 0}^L $, where the vector $|v_0, \ldots, v_N\rangle $ satisfies $A|v_0, \ldots, v_N\rangle =0$. If all the exponents $l_0, \ldots, l_N$ are non-negative, then $L=\infty $, and otherwise $L= \min_{i\in N_n} \{-\frac{v_i}{l_i}\}$.

Further, in order to simplify our investigations, we will consider only the case when $g_0$ is a constant function.

 Let us define the complex analytic map $K_{c_1, \ldots, c_N} : \mathbb{C} \ni z \mapsto |z; c_1, \ldots, c_N \rangle \in \mathcal{H}_{c_1, \ldots, c_N}$ of the complex plane $\mathbb{C}$ into Hilbert space $\mathcal{H}_{c_1, \ldots, c_N}$ by
 \begin{multline}\label{nowystan}
|z; c_1, \ldots, c_N> := \\
=\sum_{n=0}^L \frac{ z^n }{(\hbar^{\frac{1}{2}(l_0+\ldots +l_N)}g_0)^n \sqrt{(v_0+nl_0)! \ldots (v_N+nl_N)!}} |v_0+nl_0, \ldots, v_N+nl_N\rangle ,
\end{multline}
where $L+1 = \dim \mathcal{H}_{c_1, \ldots, c_N}$. 
From \eqref{Glauberstate} and \eqref{nowystan} we have 
\begin{multline}\label{withfactor}
P_{c_1, \ldots, c_N} |z_0, \ldots, z_N \rangle = \frac{z_0^{v_0} \ldots z_N^{v_N}}{\sqrt{\hbar^{v_0+\ldots +v_N}}}|z ; c_1, \ldots, c_N\rangle =\\
= \frac{e^{i\sum_{j=0}^N \frac{c_j}{\hbar}\psi_j}}{\sqrt{\hbar^{v_0+\ldots + v_N}}} \left(\sum_{j=0}^N \kappa_{0j}I_j\right)^{\frac{v_0}{2}} \ldots \left(\sum_{j=0}^N \kappa_{Nj}I_j\right)^{\frac{v_N}{2}}   |z ; c_1, \ldots, c_N\rangle ,
 \end{multline}
 where 
 \begin{multline}\label{cvz}
 z = g_0 \frac{\prod_{i\in N_p} z_i^{l_i}}{\prod_{j\in N_n} z_j^{|l_j|}} = \frac{1}{\prod_{j\in N_n} \left(\sum_{k=0}^N \kappa_{jk} I_k\right)^{|l_j|}} \sqrt{\mathcal{G}_0 (I_0, \ldots, I_N) }e^{i\psi_0}=\\
= : \sqrt{ \tilde{\mathcal{G}}_0 (I_0, \ldots, I_N) }e^{i\psi_0} ,
 \end{multline}
$\mathcal{G}_0 (I_0, \ldots, I_N)$ is defined by \eqref{funkcjagiezero} for $g_0 = const $ and $c_0, \ldots, c_N$ are related to $v_0, \ldots, v_N$ by \eqref{lambdai}.



From \eqref{cvz} and \eqref{kapa} we obtain the dependences 
\begin{equation}\label{nowezet}
|z|^2 = g_0^2 \frac{\prod_{i\in N_p} \left(l_iI_0 + \sum_{k=1}^N \kappa_{ik} I_k\right)^{l_i} }{\prod_{j\in N_n } \left(l_j I_0 + \sum_{k=1}^N \kappa_{jk} I_k\right)^{|l_j|}} = \tilde{\mathcal{G}}_0 (I_0, \ldots, I_N )
\end{equation}
and 
\begin{equation}\label{4.21}
\mathcal{G}_0(I_0, \ldots , I_N ) = \prod_{j\in N_n} \left(l_j I_0 + \sum_{k=1}^N \kappa_{jk} I_k\right)^{2|l_j|} \tilde{\mathcal{G}}_0 (I_0, \ldots, I_N).
\end{equation}

So, substituting $\mathcal{G}_0$ and $\tilde{\mathcal{G}}_0$ into \eqref{casimir} we obtain the corresponding Kummer shapes $\mathcal{C}^{-1}(0)$ and $\tilde{\mathcal{C}}^{-1}(0)$. We also mention that since of \eqref{4.21} one has the diffeomorphism $\textbf{I}_{c_1, \ldots, c_N} : \mathcal{C}^{-1} (0) \to \tilde{\mathcal{C}}^{-1} (0)$ defined by 
\begin{equation}
\textbf{I}_{c_1, \ldots, c_N} : \mathcal{C}^{-1}(0) \ni \left(\begin{array}{l}
z \\
I_0
\end{array}\right) \mapsto \left(\begin{array}{c}
\left(\prod_{j\in N_n} \left( l_j I_0+ \sum_{k=1}^N \kappa_{jk} c_k\right)^{|l_j|}\right)^{-1} z \\
I_0 
\end{array}\right) \in \tilde{\mathcal{C}}^{-1}(0). 
\end{equation}

\begin{prop}\label{prop3}
If $I_0$ satisfies \eqref{dots}
then 
\begin{equation}\label{posgie0}
\frac{\partial \tilde{\mathcal{G}}_0 }{\partial I_0} (I_0, \ldots, I_N) > 0 .
\end{equation}
Additionally one has 
\begin{align}
\label{g0wa}
\tilde{\mathcal{G}}_0 (a, I_1, \ldots, I_N ) = 0, \\
\label{g0wb}
\tilde{\mathcal{G}}_0 (b, I_1, \ldots, I_N ) = +\infty , 
\end{align}
for $a= \max\limits_{i\in N_p} \left\{-\frac{1}{l_i} \sum_{j=1}^N \kappa_{ij} I_j\right\}\mbox{ }$ and $\mbox{ } b= \min\limits_{i\in N_n} \left\{-\frac{1}{l_i} \sum_{j=1}^N \kappa_{ij} I_j\right\}$.
\end{prop}
\begin{proof}
Denoting by $\delta_i = -\frac{1}{l_i} \sum_{j=1}^N \kappa_{ij}I_j$ one has 
\begin{equation}\label{g0zdelta}
\tilde{\mathcal{G}}_0 (I_0, \ldots , I_N ) = g_0^2\frac{\prod_{i\in N_p } l_i^{l_i}(I_0- \delta_i)^{l_i}}{\prod_{j\in N_n} l_i^{|l_i|} (I_0 - \delta_j )^{|l_j|}}. 
\end{equation}
Then one can easily see that the derivative 
\begin{equation}
\frac{\partial \tilde{\mathcal{G}}_0}{\partial I_0} (I_0, \ldots, I_N ) = g_0^2 \frac{\prod_{i\in N_p } l_i^{l_i}(I_0- \delta_i)^{l_i}}{\prod_{j\in N_n} l_i^{|l_i|} (I_0 - \delta_j )^{|l_j|}} \left(\sum_{i\in N_p } \frac{l_i }{I_0 - \delta_i } + \sum_{j\in N_n} \frac{|l_i|}{\delta_i - I_0}\right) 
\end{equation}
is positive since $a= \max_{i\in N_p} \{\delta_i\} < I_0 < \min_{i\in N_n} \{\delta_i \} =b $.

Substituting $I_0=a$, $I_0=b$ into \eqref{g0zdelta} one gets \eqref{g0wa} and \eqref{g0wb}, respectively. 
\end{proof}

From this proposition we conclude that $\tilde{\mathcal{G}}_0$ is invertible as a function of $I_0$ and thus, $I_0 = \tilde{\mathcal{G}}^{-1}_0 (|z|^2, c_1, \ldots, c_N)$. Because of this the projection $pr_1: \tilde{\mathcal{C}}^{-1}(0) \subset \mathbb{C}\times ]a, b[ \to \mathbb{C}$ on the complex plane $\mathbb{C}$ defines the diffeomorphism $\Pi: \tilde{\mathcal{C}}^{-1}(0) \to \mathbb{C} $. The interval $]a,b[$ is defined like in Proposition \ref{prop3}, where $c_1, \ldots, c_N$ are substituted instead of $I_1, \ldots, I_N$. 

The intersections of the Kummer shape $\mathcal{C}^{-1}(0)$ and the "new Kummer shape" $\tilde{\mathcal{C}}^{-1}(0)$ with the plane spanned by the axises $I_0$ and $x$ is illustrated in Figure 1. \\
\begin{figure}[h]
\includegraphics[width=12cm]{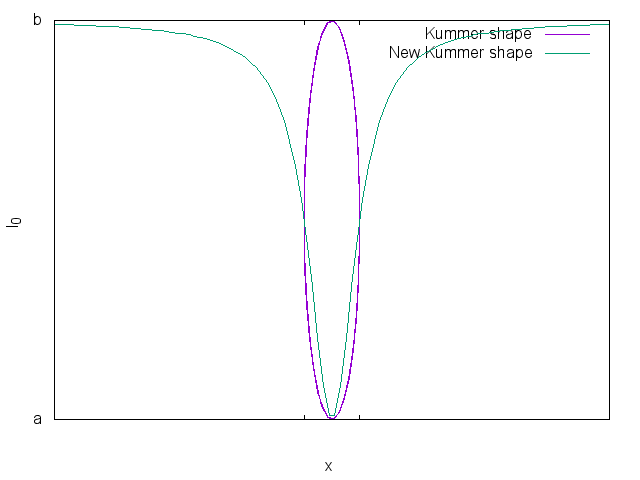}
\end{figure}
\textit{Figure 1.}

Summing up the above facts we have 

\begin{prop}
The map 
\begin{equation}
K_{c_1, \ldots, c_N} \circ \Pi \circ \textbf{I}_{c_1, \ldots, c_N} \circ \Phi_{c_1, \ldots, c_N} : \textbf{J}^{-1}(c_1, \ldots, c_N )/ \mathbb{T}^N \to \mathcal{H}_{c_1, \ldots, c_N}
\end{equation}
taken up  to a complex factor defines the reduced coherent state map \eqref{cohstat}.
\end{prop}



Let us now back to the general case in which the structural function $\mathcal{G}_\hbar$ is given by \eqref{ghf}, i.e. $g_0$ is any positive function.

One easily verifies that if $\dim_{\mathbb{C}} \mathcal{H}_{c_1, \ldots, c_N} = \infty$, then the coherent states \eqref{nowystan} can be generalized in the following way 
\begin{multline}\label{states}
|z; c_1, \ldots, c_N\rangle = \frac{1}{\sqrt{v_0!\ldots v_N!}}\bigg(|v_0, \ldots , v_N \rangle  + \\
+\sum_{n=1}^\infty \frac{z^n}{\sqrt{\mathcal{G}_\hbar (0) \ldots \mathcal{G}_\hbar (n-1)}} |v_0+nl_0, \ldots, v_N+nl_N \rangle \bigg) .
\end{multline}
Here, in order to simplify the notation, we put $\mathcal{G}_\hbar (n)$ instead of 
\begin{equation}
 \mathcal{G}_\hbar (c_0+\hbar n,  c_1 , \ldots ,  c_N ) = g_0^2(\hbar (v_0+nl_0), \ldots, \hbar (v_N +n l_N)) \mathcal{P}_{l_0}(\hbar (v_0+nl_0)) \ldots \mathcal{P}_{l_N} (\hbar (v_N + nl_N)),
\end{equation}
 where 
\begin{equation}
\mathcal{P}_{l_i}(\hbar (v_i + nl_i))=\left\{\begin{array}{ll}

\hbar^{l_i} (v_i+nl_i+1)_{l_i}& \textrm{if }l_i > 0  \\ 

1 & \textrm{if }l_i=0 \\

 \hbar^{|l_i|}(-1)^{l_i}(-v_i -nl_i)_{|l_i|} & \textrm{if }l_i<0

                      \end{array}\right.
\end{equation}
Equivalently the coherent states \eqref{states} could be defined as eigenvectors of the annihilation operator $\textbf{A}$, i.e.
\begin{equation}\label{state}
\textbf{A}|z; c_1, \ldots, c_N\rangle = z|z; c_1, \ldots, c_N\rangle . 
\end{equation}

From \eqref{states} we find that $\langle z; c_1, \ldots, c_N | z; c_1, \ldots, c_N \rangle  < \infty $ if $z\in \mathbb{D}_{R_\hbar}:=\{z\in \mathbb{C} :|z| < R_\hbar :=\limsup_{n\to \infty} \mathcal{G}_\hbar (n)\}$. In the case when function $g_0$ is constant one has $R_\hbar= \infty $. Notice also that the coherent states \eqref{states} 
form an over-complete set of vectors in $\mathcal{H}_{c_1, ..., c_N}$ when $z$ runs over the disc $\mathbb{D}_{R_\hbar}$.

In the subsequent considerations we will postulate the existence of the resolution 
\begin{equation}\label{residentity}
\mathbbm{1}_{c_1, \ldots, c_N} = \int_{\mathbb{C}} |z; c_1, \ldots, c_N \rangle \langle z; c_1, \ldots, c_N| d\mu_{c_1, \ldots, c_N} (\bar z , z)
\end{equation}
for $\mathbbm{1}_{c_1, \ldots, c_N} = P_{c_1, \ldots , c_N}|_{\mathcal{H}_{c_1, \ldots, c_N}}$ with respect to some measure $d\mu_{c_1, \ldots, c_N} (\bar z , z).$

Similarly to \eqref{covf} we define the covariant symbol 
\begin{equation}\label{covsymb}
\langle \textbf{F} \rangle (\bar z, z) := \frac{\langle z; c_1, \ldots, c_N|\textbf{F}|z; c_1, \ldots , c_N \rangle }{\langle z; c_1, \ldots, c_N|z; c_1, \ldots, c_N\rangle }, 
\end{equation}
 for an operator $\textbf{F} : \mathcal{D}(\textbf{F}) \to \mathcal{H}_{c_1, \ldots , c_N }$, if its domain $\mathcal{D}(\textbf{F})\subset \mathcal{H}_{c_1, \ldots , c_N}$ contains all coherent states defined in \eqref{states}.

Since one has one-to-one correspondence between the operators $\textbf{F},\textbf{G}$ and their symbols, we can define the $*_\hbar $-product 
\begin{equation}\label{cproduct}
(\langle \textbf{F}\rangle *_\hbar \langle \textbf{G}\rangle )(\bar z , z):= \frac{\langle z; c_1, \ldots, c_N |\textbf{FG}|z; c_1, \ldots, c_N \rangle }{\langle z; c_1, \ldots, c_N|z; c_1, \ldots, c_N\rangle },
\end{equation}
of symbols. Everywhere below we will assume that operators have a common invariant domain $\mathcal{D}$, which contains the set of all coherent states. The dependence of star product on $\hbar$ is through the coherent states.

From \eqref{cproduct} and \eqref{residentity} we obtain 
\begin{multline}\label{product2}
(\langle \textbf{F} \rangle *_\hbar \langle \textbf{G}\rangle )(\bar z , z)= \int_{\mathbb{D}_{R_\hbar} }\frac{\langle z;c_1, \ldots, c_N\mid \textbf{F}\mid w; c_1, \ldots, c_N\rangle}{\langle z; c_1, \ldots, c_N\mid w; c_1, \ldots, c_N \rangle } \times \\
\times \frac{\langle w; c_1, \ldots, c_N\mid \textbf{G}\mid z; c_1, \ldots , c_N \rangle }{\langle w; c_1, \ldots, c_N \mid z; c_1, \ldots, c_N\rangle } |a_{c_1, \ldots, c_N} (z,w)|^2 d\nu_{c_1, \ldots, c_N} (\bar w, w) ,
\end{multline}
where $ a_{c_1, \ldots, c_N} (z,w) $ is transition amplitude between the coherent states $\mathbb{C}|z; c_1, \ldots, c_N\rangle $ and $\mathbb{C}|w; c_1, \ldots , c_N\rangle$ defined by
\begin{equation}\label{amplitude}
a_{c_1, \ldots, c_N} (z,w)  := \frac{ \langle z; c_1, \ldots, c_N|w;c_1, \ldots, c_N\rangle }{\langle z; c_1, \ldots, c_N|z ; c_1, \ldots , c_N\rangle^{\frac{1}{2}} \langle w; c_1, \ldots, c_N|w; c_1, \ldots, c_N\rangle^{\frac{1}{2}} } 
\end{equation} 
and 
\begin{equation}
d\nu_{c_1, \ldots, c_N} (\bar w, w)= \langle w; c_1, \ldots, c_N | w; c_1, \ldots, c_N \rangle d\mu_{c_1, \ldots, c_N} (\bar w , w).
\end{equation} 

From \eqref{state} we see that the covariant symbols of the operators $\textbf{F}$ and $\textbf{G}$ defined by 
\begin{equation}
 \textbf{F}:=  \sum_{k,l=0}^\infty f_{\bar k , l } \textbf{A}^{*k} \textbf{A}^l  \mbox{ and } \textbf{G}:=  \sum_{r,s=0}^\infty g_{\bar r , s } \textbf{A}^{*r} \textbf{A}^s   .
 \end{equation}
 are the following 
 \begin{equation}\label{covariantsymbols}
 f(\bar z , z ) = \sum_{k,l=0 }^\infty  f_{\bar k , l } \bar z^k z^l  \quad \mbox{ and } \quad g(\bar z , z ) = \sum_{k,l =0}^\infty g_{\bar k , l } \bar z^k z^l . 
 \end{equation}
 The $*_\hbar$-product $f*_\hbar g$ of these covariant symbols is presented below in \eqref{*product}. See also \cite{AO1}. 
 \begin{prop}
 The $*_\hbar $-product \eqref{cproduct} of the covariant symbols \eqref{covariantsymbols} is given by 
 \begin{multline}\label{*product}
(f *_\hbar g) (\bar z, z )  = \\
=\frac{1}{\langle z; c_1, \ldots, c_N | z ; c_1, \ldots, c_N\rangle }  f(\bar z , \bar \partial_{\mathcal{G}_\hbar} )\left( g (\bar z , z )  \langle z; c_1, \ldots, c_N | z; c_1, \ldots , c_N \rangle \right),
\end{multline}
where, by definition, the operator $\partial_{\mathcal{G}_\hbar}$ acts on the monomial $z^n$ in the following way
\begin{equation}
\partial_{\mathcal{G}_\hbar} z^n := \mathcal{G}_\hbar (n-1) z^{n-1} , 
\end{equation}
if $n\geq 1 $ and $\partial_{\mathcal{G}_\hbar} z^n =0$ if $n=0$. The operator $f(\bar z , \bar \partial_{\mathcal{G}_\hbar })$ is defined by 
\begin{equation}
f(\bar z , \bar \partial_{\mathcal{G}_\hbar }) := \sum_{k,l=0 }^\infty  f_{\bar k , l } \bar z^k \bar \partial^l_{\mathcal{G}_\hbar} 
\end{equation}
and acts on the complex coordinate $\bar z $ only.  
\end{prop}
 \begin{proof}
In order to prove \eqref{*product} we note that from \eqref{states} and \eqref{state} it follows that 
\begin{equation}\label{expdr}
\partial_{\mathcal{G}_\hbar} \langle w; c_1. \ldots, c_N | z; c_1, \ldots, c_N \rangle = \bar w  \langle w; c_1, \ldots, c_N | z; c_1, \ldots, c_N \rangle .   
\end{equation}
Using \eqref{state} and \eqref{expdr} we have 
\begin{multline}
\langle z; c_1, \ldots, c_N|z; c_1, \ldots, c_N\rangle (f *_\hbar g) (\bar z, z )= \langle z; c_1, \ldots, c_N |\textbf{FG}|z; c_1, \ldots, c_N \rangle  = \\
=  \sum_{k,l,r,s=0}^\infty f_{\bar k , l }  g_{\bar r , s } \langle z; c_1, \ldots, c_N | \textbf{A}^{*k} \textbf{A}^l\textbf{A}^{*r} \textbf{A}^s | z; c_1, \ldots, c_N \rangle =\\
= \sum_{k,l,r,s=0}^\infty f_{\bar k , l }   g_{\bar r , s } \bar z^k  z^s \bar \partial^l_{\mathcal{G}_\hbar } \partial^r_{\mathcal{G}_\hbar} \langle z; c_1, \ldots, c_N |  z; c_1, \ldots, c_N \rangle =\\
= \sum_{k,l,r,s=0}^\infty f_{\bar k , l } \bar z^k \bar \partial^l_{\mathcal{G}_\hbar }  g_{\bar r , s }   z^s  \bar z^r \langle z; c_1, \ldots, c_N |  z; c_1, \ldots, c_N \rangle =\\
=   f(\bar z , \bar \partial_{\mathcal{G}_\hbar} ) \left(g (\bar z , z )  \langle z; c_1, \ldots, c_N | z; c_1, \ldots , c_N \rangle \right) . 
\end{multline}

\end{proof}

Now using the $*_\hbar $-product \eqref{cproduct} we will show that the quantum Kummer shape algebra corresponds to the classical case when $\hbar \to 0$. For this purpose let us formulate the following 
\begin{prop}\label{prop2}
Let us assume that the quantum structural function $\mathcal{G}_\hbar (\cdot, c_1, \ldots, c_N)$ is invertible. Then in the case when exponents $l_0, \ldots, l_N$ are non-negative one has: \\
(i)
\begin{equation}\label{firsteq} 
 \lim_{\hbar \to 0} \langle \varphi (\textbf{A}^*\textbf{A} ) \rangle = \lim_{\hbar \to 0} \langle \varphi (\textbf{A}\textbf{A}^* ) \rangle  = \varphi( \bar z z) , 
\end{equation}
where $\varphi$ is such analytic function for which the operators $\varphi (\textbf{A}^* \textbf{A} )$ and $\varphi (\textbf{A}\textbf{A}^* )$ are correctly defined. \\
(ii)
\begin{equation}\label{secondeq}
 \langle \textbf{A}_0 \rangle  = \lim_{\hbar \to 0 } \langle \left( \mathcal{G}_\hbar^{-1} (\textbf{A} \textbf{A}^*, c_1, \ldots, c_N)\right)^n \rangle = \left(\mathcal{G}_0^{-1}( \bar z z, c_1, \ldots, c_N ) \right)^n = I_0^n, 
\end{equation}
where $\mathcal{G}_{\hbar}^{-1} ( \cdot , c_1, \ldots, c_N) $ and $\mathcal{G}_0^{-1} ( \cdot , c_1, \ldots, c_N)$ are functions inverse to the structural functions $\mathcal{G}_\hbar ( \cdot , c_1, \ldots, c_N ) $ and $\mathcal{G}_0 ( \cdot , c_1, \ldots, c_N ) $, respectively. 
\end{prop}
\begin{proof}
(i) Since 
\begin{equation}
\lim_{\hbar \to 0} \langle \left(\mathcal{G}_\hbar ( \textbf{A}_0 , c_1, \ldots, c_N )\right)^n \rangle  = \lim_{\hbar \to 0} \langle \left(\mathcal{G}_\hbar (\textbf{A}_0 - \hbar , c_1, \ldots, c_N )\right)^n \rangle , 
\end{equation}
from \eqref{rel3} and \eqref{rel2} we have 
\begin{equation}
\lim_{\hbar \to 0 } \langle (\textbf{A}\textbf{A}^* )^n \rangle = \lim_{\hbar \to 0 } \langle (\textbf{A}^* \textbf{A} )^n \rangle ,
\end{equation}
for $n\in \mathbb{N}$. Let us assume that 
\begin{equation}\label{aastarlimit}
\lim_{\hbar \to 0 }  \langle (\textbf{A}\textbf{A}^* )^n \rangle = (\bar z z)^n . 
\end{equation}
Then
\begin{equation}
\langle (\textbf{A}^* \textbf{A} )^{n+1} \rangle = \bar z z \langle (\textbf{A}\textbf{A}^* )^n \rangle 
\end{equation}
and \eqref{aastarlimit} imply that 
\begin{equation}
\lim_{\hbar \to 0 } \langle (\textbf{A}^* \textbf{A} )^{n+1} \rangle  = \lim_{\hbar \to 0 } \langle (\textbf{A}\textbf{A}^* )^{n+1} \rangle = (\bar z z )^{n+1} . 
\end{equation}
So, we obtain \eqref{firsteq} by linearity. \\
(ii) In order to obtain \eqref{secondeq} we apply \eqref{firsteq} to $\mathcal{G}_\hbar^{-1}$ and note that $\lim_{ \hbar \to 0 } \mathcal{G}_\hbar = \mathcal{G}_0$. 

 \end{proof}

Let us note here that if $g_0$ is a constant, then the assumptions of the above proposition are satisfied.

One defines the Lie bracket $\{f, g\}_{\mathcal{G}_0}$ of the covariant symbols $f$ and $g$ defined in \eqref{covariantsymbols} by 
\begin{equation}\label{qpb}
\{ f, g \}_{\mathcal{G}_0} := \lim_{\hbar \to 0} \frac{-i}{\hbar } (f *_\hbar g - g *_\hbar f ).
 \end{equation}
 From Proposition \ref{prop2} we find that the covariant symbols of $\textbf{A} , \textbf{A}^* $ and $\textbf{A}_0$ are $z, \bar z $ and $I_0$, respectively. Thus, using the relations (\ref{rel1}-\ref{rel2}), we find the relations 
\begin{equation}\label{qpbx0x}
\{ I_0, z \}_{\mathcal{G}_0} = iz,
\end{equation}
\begin{equation}
\{I_0, \bar z \}_{\mathcal{G}_0} =  -i\bar z,
\end{equation}
\begin{equation}\label{qpbxy}
\{z, \bar z\}_{\mathcal{G}_0} =  -i\frac{\partial \mathcal{G}_0}{\partial I_0} (I_0),
\end{equation}
which are the reduced versions of the ones presented in (\ref{pbx0x}-\ref{pbxy}).

From (\ref{qpbx0x}-\ref{qpbxy}) we see that the bracket \eqref{qpb} is the Nambu bracket defined in (\ref{pbx0x}-\ref{pbxy}). So, $\{ f, g \}_{\mathcal{G}_0}$ is the Poisson bracket of the covariant symbols $f$ and $g $ of the operators $\textbf{F} $ and $ \textbf{G} $ belonging to the quantum Kummer shape algebra $\mathcal{Q}_{\mathcal{G}_\hbar } (\mathcal{H}_{c_1 , \ldots , c_N })$.

Therefore, in the classical limit $\hbar \to 0$ the quantum Kummer shape, i.e. the operator algebra defined by the operators (\ref{rel1}-\ref{rel2}), corresponds to the classical Kummer shape (\ref{pbx0x}-\ref{pbxy}) with Nambu bracket $\{\cdot , \cdot \}_{\mathcal{C}}$ defined by structural function $\mathcal{G}_0$, see \eqref{casimir} and \eqref{nambu}.

We summarize our considerations in the following proposition. 

\begin{prop}
The passing to the classical limit  $\hbar \to 0 $ intertwines both reduction procedures, quantum and classical, i.e.

Figure 2. 

\begin{center}
\xymatrix @C=2cm @R=3cm {
*+[F] \txt{Quantum system of $N+1$ \\ harmonic oscillators} \ar[d]^{\txt{Quantum \\ reduction}}\ar[r]^{\hbar \to 0} &*+[F]\txt{Classical system of $N+1$\\ harmonic oscillators}\ar[d]^{\txt{Classical\\ reduction}} \\
*+[F]\txt{Quantum Kummer shape}\ar[r]^{\hbar \to 0} & *+[F]\txt{Classical Kummer shape}
}
\end{center}
\end{prop}

\section{Reproducing measure for the reduced coherent states}\label{sec:7}

In the previous section in order to construct the coherent state map for a classical Kummer shape we have defined, see \eqref{nowystan}, the complex analytic map $K_{c_1, \ldots, c_N} : \mathbb{C} \to \mathcal{H}_{c_1, \ldots, c_N}$ of $\mathbb{C}$ into reduced Hilbert space $\mathcal{H}_{c_1, \ldots, c_N}$. Now we will derive a measure $d\mu_{c_1, \ldots, c_N} (\bar z , z) $ which gives the resolution \eqref{residentity} of identity for the coherent states map \eqref{nowystan}.

 Equivalently one can rewrite \eqref{residentity} as 
\begin{equation}\label{repproperty}
\psi (w) = \int_{\mathbb{C}} \psi (z) K_{c_1, \ldots, c_N} ( \bar z , w ) d\mu_{c_1, \ldots, c_N}  (\bar z , z ) , 
\end{equation}
where 
\begin{equation}\label{caf}
\psi (w) :=  \langle \psi | w; c_1, \ldots, c_N \rangle 
\end{equation}
for $|\psi \rangle \in \mathcal{H}_{c_1, \ldots, c_N}$, and 
\begin{equation}\label{kernelog}
K_{c_1, \ldots, c_N} ( \bar z , w ) = \langle z ; c_1, \ldots, c_N | w; c_1, \ldots, c_N \rangle . 
\end{equation}
 Therefore, \eqref{kernelog} is a reproducing kernel for the complex analytic functions defined by \eqref{caf}. One can express this kernel in terms of the hypergeometric function $_rF_s \left[\begin{array}{cccc} 
\alpha_1, & \alpha_2, & \ldots, & \alpha_r \\
\beta_1, & \beta_2, & \ldots , & \beta_s \end{array} ; \cdot \right]$ , where $r = 1 + \sum_{i\in N_n} |l_i|$ and $s= \sum_{i\in N_p} l_i $. Namely, using the equalities 
\begin{equation} 
\hbar^{v_i+nl_i} (v_i+nl_i)! = \left\{\begin{array}{ll}
\hbar^{v_i} v_i! \mathcal{P}_{l_{i}}(\hbar v_{i}) \ldots \mathcal{P}_{l_{i}}(\hbar v_{i} + \hbar l_{i} (n-1)), & l_i \geq 0 \\
\frac{\hbar^{v_i} v_i!}{ \mathcal{P}_{l_{i}}(\hbar v_{i}) \ldots \mathcal{P}_{l_{i}}(\hbar v_{i} + \hbar l_{i} (n-1))}, & l_i <0
\end{array}\right.
\end{equation}
we find that
\begin{multline}\label{kernelseries}
K_{c_1, \ldots, c_N} (\bar z , w)  = \frac{1}{v_0!\ldots v_N!} \mbox{ } _rF_s\left[\begin{array}{cccc} 
\alpha_1, & \alpha_2, & \ldots, & \alpha_r \\
\beta_1, & \beta_2, & \ldots , & \beta_s \end{array} ; \frac{\bar z w }{l_0^{l_0} \ldots l_N^{l_N} \hbar^{l_0+ \ldots + l_N} }\right] ,
\end{multline} 
where 
\begin{equation}
(\alpha_1 , \alpha_2 , \ldots, \alpha_r ) = \left(1, \frac{v_{i_1}}{l_{i_1} }, \ldots , \frac{v_{i_1}-(-l_{i_1}-1)}{l_{i_1}}, \ldots , \frac{v_{i_{r-1}}}{l_{i_{r-1}}},  \ldots , \frac{v_{i_{r-1}}-(-l_{i_{r-1}}-1)}{l_{i_{r-1}}}\right)
\end{equation}
and
\begin{equation}
(\beta_1, \beta_2, \ldots, \beta_s ) = \left(\frac{v_{j_1}+1}{l_{j_1} }, \ldots , \frac{v_{j_1}+l_{j_1}}{l_{j_1}}, \ldots,  \frac{v_{j_{s}}+1}{l_{j_{s}}},  \ldots , \frac{v_{j_{s}}+l_{j_{s}}}{l_{j_{s}}} \right).
\end{equation}

Note here that in the case when $N_n \neq \emptyset $, i.e. when $\dim_{\mathbb{C}} \mathcal{H}_{c_1, \ldots, c_N } < +\infty $, the kernel \eqref{kernelseries} is a polynomial function of the variable $\bar z w$. 

\begin{prop}\label{prop4}
The reproducing measure in \eqref{repproperty} has the following form 
\begin{equation}\label{densityfunction}
d\mu_{c_1, \ldots, c_N} ( \bar z , z ) = \rho_{c_1, \ldots, c_N} (|z|^2)   d |z|^2d\psi_0 = i \rho_{c_1, \ldots, c_N} (|z|^2)dz d\bar z, 
\end{equation}
where $z= |z|e^{i\psi_0}$ and 
\begin{multline}\label{rofunkcja}
\rho_{c_1, \ldots, c_N} (|z|^2  ) := \frac{1}{2\pi l_0^2 (\hbar )^{N+1} }|z|^{2\left(\frac{v_0+1}{l_0}-1\right)}\times\\
 \times \int_{[0, +\infty )^N} x_1^{v_1 -\frac{l_1(v_0+1)}{l_0}} \ldots x_N^{v_N - \frac{l_N(v_0+1)}{l_0}}e^{-\frac{1}{\hbar}(|z|^{\frac{2}{l_0}} (x_1^{l_1} \ldots x_N^{l_N} )^{\frac{1}{l_0}}+x_1 +\ldots +x_N } dx_1 \ldots dx_N.
\end{multline}
\end{prop}
\begin{proof}
Applying the projection $P_{c_1, \ldots, c_N} $ on the both sides of \eqref{glauber1} we find that 
\begin{multline}\label{dziesiec}
\mathbbm{1}_{c_1, \ldots, c_N} = \int_{\mathbb{C}^{N+1}} \frac{ P_{c_1, \ldots, c_N}|z_0, ... , z_N \rangle \langle z_0, ... , z_N |P_{c_1, \ldots, c_N} } {\langle z_0, ... , z_N | z_0, ... , z_N \rangle } d\nu_\hbar (\bar z_0,..., \bar z_N, z_0,..., z_N )  = \\
=  \int_{\mathbb{C}^{N+1}} |z;c_1, \ldots, c_N\rangle \langle z; c_1, \ldots, c_N |  \frac{|z_0|^{2v_0} \ldots |z_N|^{2v_N} }{\hbar^{v_0+\ldots +v_N}}\times \\
\times e^{\frac{1}{\hbar}(|z_0|^2 + \ldots + |z_N|^2)}d\nu_\hbar (\bar z_0,..., \bar z_N, z_0,..., z_N ). 
\end{multline}
where the second equality in \eqref{dziesiec} follows from \eqref{withfactor}. After passing to the polar coordinates we obtain 
\begin{multline}
\mathbbm{1}_{c_1, \ldots, c_N}= \frac{1}{\hbar^{v_0 + \ldots + v_N} (2\pi \hbar )^{N+1}} \times \\
\times \int_{\Omega_{N+1}} |z;c_1, \ldots, c_N\rangle \langle z; c_1, \ldots, c_N | x_0^{v_0} \ldots x_N^{v_N} e^{\frac{1}{\hbar} (x_0 + \ldots + x_N)} dx_0 \wedge d\phi_0 \wedge \ldots \wedge dx_N \wedge d\phi_N , 
\end{multline}
where $x_k = |z_k|^2$. From \eqref{psizerodef} and \eqref{cvz} we have 
\begin{align}
\notag 
\phi_0 & = \frac{1}{l_0}\left( \psi_0 - \sum_{j=1}^N l_j \phi_j\right) , \\
x_0 & = x^{\frac{1}{l_0}} (x_1^{l_1} \ldots x_N^{l_N})^{-\frac{1}{l_0}} , 
\end{align}
where $x:= |z|^2$. Thus, one has 
\begin{multline}
dx_0 \wedge d\phi_0 \wedge dx_1 \wedge d\phi_1 \wedge \ldots \wedge dx_N \wedge d\phi_N = \\
=\frac{1}{l_0^2} x^{\frac{1}{l_0}-1} (x_1^{l_1} \ldots x_N^{l_N})^{-\frac{1}{l_0}} dx \wedge d\psi_0 \wedge dx_1 \wedge d\phi_1 \wedge \ldots \wedge dx_N \wedge d\phi_N  . 
\end{multline}
Changing in \eqref{dziesiec}  the variables $(x_0, \phi_0)$ on the variables $(x, \psi_0 )$ we find that the density function in \eqref{densityfunction} is given by \eqref{rofunkcja}. Note here that $dx \wedge d\psi_0 = i dz \wedge d\bar z$. 
\end{proof}

In such a way we have found the family of reproducing kernels indexed by the vacuum states $|c_0, \ldots, c_N \rangle $, see \eqref{v1} or \eqref{vacuum}, where $(c_0, \ldots, c_N)^T \in \mathbb{R}^{N+1}$ are related to $(v_0, \ldots, v_N)^T \in \mathbb{Z}^{N+1}$ by \eqref{lambdai}. The reproducing measures for this kernels are given explicity by the integral formula \eqref{rofunkcja}. The examples presented below shows that this family of reproducing kernels includes the known reproducing kernels as well as the new ones. Let us mention also that for the exponents $(l_0, \ldots, l_N)^T \in \mathbb{Z}^{N+1}$ assuming positive and negative values we obtain the reproducing kernels in the finite dimensional Hilbert spaces realized by the polynomials of the complex variable $z \in \mathbb{C}$.

In the case $N=1$ if $l_0, l_1$ are non-negative one has
\begin{equation}\label{elpositive}
K_{c_1} ( \bar z , w) = \frac{1}{v_0! v_1!}\mbox{ } _1F_{l_0+l_1} \left[\begin{array}{llllll}
 & &1 & & & \\
\frac{v_0+1}{l_0} ,& \ldots, &\frac{v_0+l_0}{l_0},&\frac{v_1+1}{l_1}, &\ldots, & \frac{v_1+l_1}{l_1}
\end{array} ; \frac{\bar z w }{\hbar^{n(l_0+l_1)} l_0^{l_0}l_1^{l_1}}\right] , 
\end{equation}
and if $l_0 \geq 0, l_1 <0 $ one has
\begin{equation}\label{eldifferent}
K_{c_1} ( \bar z , w) = \frac{1}{v_0! v_1!}\mbox{ } _{1+|l_0|}F_{l_0} \left[\begin{array}{llll}
1 & \frac{v_1}{l_1}, &\ldots  ,& \frac{v_1-(-l_1-1)}{l_1}  \\
 & \frac{v_0+1}{l_0} ,& \ldots, &\frac{v_0+l_0}{l_0}
\end{array} ; \frac{\bar z w }{\hbar^{n(l_0+l_1)} l_0^{l_0}l_1^{l_1}}\right] .
\end{equation}
The reproducing measure for the above kernels is given by the density function
\begin{equation}\label{rodens}
\rho_{c_1} (|z|^2) = \frac{1}{2\pi \hbar^{v_0 + v_1 +2}l_0^2} |z|^{2\left(\frac{v_0+1}{l_0} -1\right)} \int_0^\infty  x_1^{v_1 - \frac{l_1(v_0+1)}{l_0}} e^{\frac{1}{\hbar} (|z|^{\frac{2}{l_0}}x_1^{-\frac{l_1}{l_0}} +x_1)} dx_1 . 
\end{equation}

At the end we present two examples. 
\begin{example}
For $l_0=l_1 = 1$ and $|0, v_1\rangle$ as a vacuum state, $v_1 \in \mathbb{N}$, the reproducing kernel \eqref{elpositive} takes the form 
\begin{equation}
K_{c_1}(\bar z , w) = \frac{1}{v_1!} \mbox{ }  _0F_1 \left[- ; v_1+1 ; \frac{\bar z w}{\hbar^2}\right] .
\end{equation}
The density function \eqref{rodens} is given by 
\begin{multline}\label{pro}
\rho_{c_1}( |z|^2 ) = \frac{1}{2\pi\hbar^{v_1+2}} \int_{0}^{+\infty} x_1^{v_1-1} e^{-\frac{1}{\hbar}\left(\frac{|z|^2}{x_1} +x_1\right)} dx_1 =\\
=\frac{1}{2\pi \hbar^2 } \left(\frac{|z|^2}{\hbar^2}\right)^{\frac{v_1}{2}} K_{v_1}\left(2\frac{|z|}{\hbar}\right) ,  
\end{multline}
where $K_{v_1}$ is the modified Bessel function of the second kind. For the last formula see entry 3.471.9 in \cite{Ryzhik}. Note that the density function \eqref{pro} is the same as the one presented in \cite{Aneta}.
\end{example}

\begin{example}
Assuming $l_0=1, l_1 = -1$ and choosing $|0, v_1 \rangle $ as the vacuum state, one can rewrite \eqref{eldifferent} in the following form
\begin{equation}
K_{c_1}(\bar z ,w) = \frac{1}{v_1!} (1 + \bar z w)^{v_1} . 
\end{equation}
From \eqref{rodens} one has 
\begin{equation}
\rho_{c_1} (|z|^2) = \frac{1}{2\pi  \hbar^{v_1+2}} \int_0^{+\infty} x_1^{v_1+1} e^{-\frac{1}{\hbar} (1+|z|^2)x_1 } dx_1 = \frac{1}{2\pi} \frac{(v_1+1)!}{(1+ |z|^2)^{v_1+2}} . 
\end{equation}
\end{example}

These examples are related to the Examples \ref{laguerre} and \ref{krawtchouk} presented in Section \ref{sec:examples}.

\section{Quantum integrability}\label{sec:integrability}

In Section \ref{sec:kummer} we have shown that the hamiltonian system of $N+1$ harmonic oscillators described by the Hamiltonian \eqref{H} is integrable in quadratures after reduction to the classical Kummer shape. The natural question arises: whether the quantum system described by the the quantum Hamiltonian \eqref{qH} is integrable in the similar way.

Let us mention here that the integration of any quantum system is equivalent to finding of spectral resolution of its Hamiltonian $\textbf{H}$. Hence, one can obtain the corresponding unitary flow $\mathbb{R} \ni t \mapsto e^{\frac{i}{\hbar} t \textbf{H}} \in \mbox{Aut} (\mathcal{H})$ in the Hilbert space $\mathcal{H}$, which defines the time evolution $ | \psi (t) \rangle = e^{\frac{i}{\hbar } t \textbf{H}} | \psi  \rangle $ and $\textbf{F}(t) = e^{-\frac{i}{\hbar } t \textbf{H}} \textbf{F} e^{\frac{i}{\hbar } t \textbf{H}}$ of a state $|\psi  \rangle \in \mathcal{H}$ and an operator $\textbf{F}$, respectively.

Treating the quantum system \eqref{qH} in the Heisenberg picture one finds that the operators $A_1, \ldots , A_N  $ form an involutive system of quantum integrals of motion with the discrete set of common eigenvalues defined in \eqref{lambdai}. The other three observables $\textbf{A}_0 (t)$, $\textbf{X}(t) := \frac{1}{2}(\textbf{A} (t) + \textbf{A}^* (t) ) $ and $\textbf{Y}(t) := \frac{1}{2i} (\textbf{A}(t) - \textbf{A}^* (t))$ satisfy the Heisenberg equations: 
\begin{equation}\label{heis1}
i\hbar \frac{d}{dt} \textbf{A}_0 (t) = [\textbf H, \textbf{A}_0 (t)]  = 2 [ \textbf{X}(t) , \textbf{A}_0 (t) ] = 2 i\hbar \textbf{Y}(t), 
\end{equation}
\begin{equation}\label{heis2}
i\hbar \frac{d}{dt} \textbf{X}(t) = [\textbf{H} , \textbf{X}(t)] = [H_0 (\textbf{A}_0(t), c_1, \ldots , c_N), \textbf{X}(t) ], 
\end{equation}
\begin{multline}\label{heis3}
i \hbar \frac{d}{dt} \textbf{Y}(t) = [\textbf{H} , \textbf{Y}(t) ] = [H_0 (\textbf{A}_0(t), c_1, \ldots , c_N), \textbf{Y}(t)] + \\
i (\mathcal{G}_\hbar (\textbf{A}_0 (t) , c_1, \ldots , c_N ) - \mathcal{G}_\hbar (\textbf{A}_0 (t) - \hbar , c_1, \ldots , c_N )), 
\end{multline}
where $\textbf{H}$ is defined by \eqref{Hr}. Additionally, from \eqref{Hr} and \eqref{heis1} one has 
\begin{equation}\label{heis}
\textbf{X}(t) = \frac{1}{2} \left[\textbf{H} - H_0 (\textbf{A}_0(t), c_1, \ldots , c_N )\right], 
\end{equation}
\begin{equation}\label{heiss}
\textbf{Y}(t) = \frac{1}{2} \frac{d}{dt} \textbf{A}_0 (t) .  
\end{equation}
Since, $\textbf{H} (t) = \textbf{H} $ is an integral of motion, the equations  \eqref{heis} and \eqref{heiss} allow us to express $\textbf{X}(t)$ and $\textbf{Y}(t)$ by $\textbf{A}_0 (t)$ which is the solution of the operator equation 
\begin{multline}\label{nazero}
\left(\frac{d}{dt} \textbf{A}_0 (t)\right)^2 =   2 [\mathcal{G}_\hbar (\textbf{A}_0 (t) , c_1, \ldots , c_N) + \mathcal{G}_\hbar (\textbf{A}_0 (t) -\hbar, c_1, \ldots , c_N ) ]- \\
\left[ \textbf{H} - H_0 (\textbf{A}_0(t) , c_1, \ldots , c_N )\right]^2. 
\end{multline}
In the case when $l_0, \ldots , l_N $ are positive integer numbers, taking mean values of both sides of Heisenberg equations \eqref{heis1} and \eqref{heis2} on the coherent states \eqref{states}, one obtains equations on the symbols $\langle \textbf{A}_0 (t) \rangle, \langle \textbf{X}(t) \rangle $ and $\langle \textbf{Y}(t) \rangle $ of the operators $\textbf{A}_0 (t), \textbf{X}(t)$ and $ \textbf{Y}(t)$, which in the limit $\hbar \to 0$ tend to the Hamilton equations (\ref{pbx0x}-\ref{pbxy}).

The above shows the correspondence between the quantum and the Hamiltonian considered within the framework of the quantum and classical Kummer shape algebras. However, unlike the reduced Heisenberg equation \eqref{nazero} its classical counterpart \eqref{difi0} is integrated in quadratures, thus allowing us to solve the Hamilton equations (\ref{i_0t}-\ref{dyt}). As we already mentioned, the solution of \eqref{nazero} is found by solving the spectral problems for the reduced Hamiltonians \eqref{Hr} indexed by the eigenvalues $c_1, \ldots , c_N$. The operators \eqref{Hr} are three-diagonal, so, one of the methods of their "diagonalization" is the given relation of the theory of Jacobi operators to the orthogonal polynomials theory, including finite orthogonal polynomials case. Recall that the last case occurs when the exponents $l_0, \ldots , l_N$ have mixed signs. The examples of integrable Hamiltonians of the form \eqref{Hr}, which describe some quantum optical processes, can be found in \cite{AO2}. Other examples are presented in the next section of the paper. 

\section{Examples}\label{sec:examples}

We will present three examples which exhibit the usefulness of Kummer shape algebra concept. In choosing these examples first of all we were  motivated by their importance in nonlinear optics (classical and quantum), e.g. see \cite{Mil}. Secondly, we want to stress the importance of quantum Kummer shape algebras for the spectral theory of Jacobi operators \cite{Tesch} and the theory of special functions \cite{AO1}, \cite{Rahman}.  

Example \ref{laguerre} and Example \ref{krawtchouk} are related to Laguerre and Krawtchouk polynomials, respectively. The corresponding Hamilton operators, see \eqref{qHr} and \eqref{qHr2}, model the interactions between two modes and describe rich variety of quantum optical phenomena since they give the possibility of correlations between these modes, see \cite{ATer}, \cite{Mil}.

Example \ref{qheisenberg} will concern the $q$-deformed Heisenberg-Weyl algebra as well as the polynomials which are a $q$-deformation of the Hermite polynomials. We chose this case in order to illustrate the connection of quantum Kummer shape algebras with the theory of $q$-special functions.

\begin{example}[\textit{Laguerre polynomials case}]\label{laguerre}

We will assume $N=1$ and $(l_0, l_1)=(1,1)$. As Hamiltonian \eqref{H} we take
\begin{equation}\label{pH}
H=|z_0|^2 + |z_1|^2 + z_0z_1 + \bar z_0 \bar z_1. 
\end{equation}

In order to define the variables $I_0, I_1, \psi_0 , \psi_1$ we choose the $\rho $ matrix in the form 
$
\rho = \frac{1}{2}\left(\begin{array}{cc} 
1 & 1 \\
-1 & 1 
\end{array}\right) , 
$ what gives 
\begin{align}
I_0= \frac{1}{2} (|z_0|^2 +|z_1|^2), & \quad \psi_0 = \phi_0 +\phi_1, \\
I_1 = \frac{1}{2} (|z_1|^2 - |z_0|^2), & \quad \psi_1 = \phi_1 -  \phi_0 ,
\end{align}
Thus, coordinates $I_0, I_1$ go through the cone  in $\mathbb{R}^2$ defined by 
\begin{align}
\notag I_0-I_1 > 0 \\
\label{pte0} I_0+I_1> 0, 
\end{align}
while for $(\psi_0, \psi_1 ) $ one has $0< \psi_0 \leq 4\pi $ and $0< \psi_1 \leq 4\pi$.

Hamiltonian \eqref{pH} in terms of the canonical coordinates $I_0, I_1, \psi_0, \psi_1$ has the form 
\begin{equation}\label{pH2}
H = 2I_0 + 2 \sqrt{(I_0-I_1)(I_0+I_1)}\cos \psi_0
\end{equation}
and Hamilton equations (\ref{eq-i0}-\ref{eq-psik}) for \eqref{pH2} are 
\begin{align}
\label{peq-i0}
 \frac{dI_0}{dt} &=2 \sqrt{(I_0-I_1)(I_0+I_1)} \sin \psi_0 \\
\label{peq-ik}
 \frac{dI_1}{dt} &=0 \\
\label{peq-psi0}
 \frac{d\psi_0 }{dt} &=2 +  \frac{2I_0}{\sqrt{(I_0-I_1)(I_0+I_1)}} \cos \psi_0  \\
\label{peq-psik}
 \frac{d\psi_1 }{dt} &=\frac{-2I_1}{\sqrt{(I_0-I_1)(I_0+I_1)}} \cos \psi_0  . 
 \end{align}

 For a fixed value $c_1\in \mathbb{R}$ of $I_1$, from \eqref{pte0} one finds that $I_0 > \max \{ -c_1, c_1\}$. So, the level $\textbf{J}^{-1}(c_1)$ of the momentum map \eqref{mmap} is isomorphic to $]|c_1|, +\infty [ \times \mathbb{S}^1 \times \mathbb{S}^1$. On the reduced phase space $\textbf{J}^{-1}(c_1)/\mathbb{S}^1 \cong ]|c_1|, +\infty [\times \mathbb{S}^1 $ Hamilton equations (\ref{peq-i0}-\ref{peq-psik}) reduce to 
 \begin{align}
 \label{p1eq-i0}
 \frac{dI_0}{dt} &=2 \sqrt{(I_0-c_1)(I_0+c_1)} \sin \psi_0 \\
\label{p1eq-psi0}
 \frac{d\psi_0 }{dt} &=2 +  \frac{2I_0}{\sqrt{(I_0-c_1)(I_0+c_1)}} \cos \psi_0  .
 \end{align}

  In considered case the classical structural function $\mathcal{G}_0$ is given by
 \begin{equation}\label{giezero}
 \mathcal{G}_0 (I_0, I_1) = (I_0-I_1)(I_0+I_1)
 \end{equation} 
 and the function $\mathcal{F}$ defined in \eqref{cfz} is 
 \begin{equation}
 z= x+iy = z_0z_1 = \sqrt{(I_0-I_1)(I_0+I_1)} e^{i\psi_0}.
 \end{equation}
 Hence, the Kummer shape is the upper half of two-sheeted hyperboloid: 
 \begin{equation}\label{twos}
 I_0^2-x^2-y^2 = c_1^2. 
 \end{equation}
 Functions $x, y$ and $I_0$ satisfy relations 
 \begin{equation}\notag
 \{ I_0, x\} = -y,
 \end{equation}
 \begin{equation}\label{6.20}
 \{ I_0 , y \} = x,
\end{equation} 
 \begin{equation}\notag
 \{ x, y\} = I_0 .
 \end{equation} 
 So, the classical Kummer shape algebra $\mathcal{K}_{\mathcal{G}_0} (c_1)$ is the Lie algebra of the group $SO(2, 1 )$.

In terms of $x,y$ and $I_0$ Hamiltonian \eqref{pH2} simplifies to 
 \begin{equation}\label{pH3}
 H= 2(I_0+x) 
\end{equation}
and the corresponding Hamilton equations are linear:
\begin{align}
\notag \frac{dI_0}{dt} &=2y \\
\label{6.23} \frac{dx}{dt} &=-2y  \\
\notag \frac{dy}{dt} &=2x  + 2I_0 . 
\end{align}
 Therefore, they can be integrated in an elementary way. One can also obtain the trajectories of \eqref{6.23} as intersections of the planes defined in \eqref{pH3} with the upper half of the two-sheeted hyperboloid \eqref{twos}.

The quantum counterpart of the Hamiltonian \eqref{pH} is
 \begin{equation}\label{pqH}
 \textbf{H} = a_0^*a_0 + a_1^*a_1  +a_0a_1 + a_0^*a_1^* +\hbar 
 \end{equation}
 and the quantum Kummer shape algebra $\mathcal{Q}_{\mathcal{G}_\hbar}(\mathcal{H}_{c_1})$ is generated by  
operators $\textbf{A}_0, \textbf{A}$ and $\textbf{A}^* $: 
\begin{equation}
\textbf{A}_0 |c_0+\hbar n , c_1\rangle = (c_0 + \hbar n ) |c_0 +\hbar n , c_1 \rangle  = (\hbar (\frac{v_1}{2}+n ))|n, v_1+n\rangle , 
\end{equation}
\begin{multline}
\textbf{A} |c_0+\hbar n , c_1 \rangle = \sqrt{(c_0-c_1 + \hbar n )(c_0+c_1+\hbar n) } |c_0+\hbar (n-1), c_1 \rangle = \\
=\hbar\sqrt{ n (v_1 + n )} |n-1, v_1 + n-1\rangle , 
\end{multline}
\begin{multline}
\textbf{A}^* |c_0+\hbar n , c_1 \rangle = \sqrt{(c_0-c_1+\hbar (n+1))(c_0+c_1 +\hbar (n+1))} |c_0 + \hbar (n+1), c_1 \rangle = \\
=\hbar\sqrt{ ( n+1)( v_1+n+1)}|n+1, v_1 +n+1 \rangle ,
\end{multline}
where we have choosen $|0, v_1 \rangle $ , $v_1 \in \mathbb{N}$, as the vacuum for the reduced Hilbert space $\mathcal{H}_{c_1}$. The quantum structural function $\mathcal{G}_\hbar $ is the following 
\begin{equation}
 \mathcal{G}_\hbar (A_0, A_1) = (A_0-A_1 +\hbar )(A_0+A_1 + \hbar)
 \end{equation}
 and in the limit $\hbar \to 0$ gives the classical structural function $\mathcal{G}_0$.
Operators $\textbf{A}_0, \textbf{A}$ and $\textbf{A}^* $ satisfy the relations 
\begin{equation}
[\textbf{A}_0, \textbf{A}] = -\hbar \textbf{A} , 
\end{equation}
\begin{equation}
[\textbf{A}_0, \textbf{A}^* ] = \hbar \textbf{A}^*,
\end{equation}
\begin{equation}\label{przykladrel}
 [\textbf{A} , \textbf{A}^*] = 2\hbar \textbf{A}_0 +\hbar^2 . 
\end{equation}
So, the quantum Kummer shape algebra is a central extension of the classical Kummer shape algebra \eqref{6.20}.

The reduced quantum Hamiltonian \eqref{Hr} in this case is 
\begin{equation}\label{qHr}
\textbf{H} = 2\textbf{A}_0 +  \textbf{A} + \textbf{A}^* + \hbar .
\end{equation}
In the basis $|n, v_1 +n \rangle $, $n\in \mathbb{N}\cup \{0\}$, the above operator has three-diagonal form which is the same as the three-term recurrence operator for the Laguerre polynomials. So, it is unitarily equivalent to the operator of multiplication by $x$ in the Hilbert space $L^2 ([0, +\infty ), x^{n_1}e^{-x} dx )$, see \cite{AO} and \cite{AO2}.

Now, having obtained classical and quantum Kummer shape algebras, we will show the correspondence between them.

Using \eqref{nowystan} and \eqref{withfactor} we find that the reduced coherent states in this case are given by 
\begin{equation}
P_{c_1} |z_0, z_1 \rangle = \frac{z_1^{v_1}}{\sqrt{\hbar^{v_1}}} |z; 0, v_1 \rangle , 
\end{equation}
where 
\begin{equation}\label{pstates}
|z; 0, v_1\rangle = \frac{1}{\sqrt{v_1!}} \left( |0, v_1 \rangle + \sum_{n=1}^\infty \frac{z^n}{\hbar^n \sqrt{n! (v_1+1)_n}}|n, v_1+n \rangle\right) 
\end{equation}
and $(a)_n := a(a+1) \ldots (a+n-1)$.

Using the coherent states \eqref{pstates} we can compute the covariant symbols of the reduced operators $\textbf{A}_0, \textbf{A}$ and $\textbf{A}^*$ obtaining
\begin{equation}
\langle \textbf{A} \rangle = z , \quad \langle \textbf{A}^* \rangle = \bar z , \quad \langle \textbf{A}_0 \rangle  \underset{\hbar \to 0 }{\longrightarrow } I_0.
\end{equation}
The Lie bracket \eqref{qpb} in the limit $ \hbar \to 0 $ satisfies the relations \eqref{6.20}.

The reduced Heisenberg equations: 
\begin{align}
\frac{d}{dt} \textbf{A}_0 (t)  & = 2 \textbf{Y}(t), \\
\frac{d}{dt} \textbf{X} (t) & = -2 \textbf{Y}(t), \\
\frac{d}{dt} \textbf{Y} (t) & = 2\textbf{X}(t) + 2 \textbf{A}_0 (t) = \textbf{H},
\end{align}
defined by the Hamiltonian \eqref{pqH} are linear, hence, they can be integrated in an elementary way. 
In the limit $\hbar \to 0$ covariant symbols of $\textbf{A}_0 (t)$,  $\textbf{X} (t) $ and $\textbf{Y} (t) $ satisfy the hamiltonian equations \eqref{6.23}. 
\end{example}

\begin{example}[\textit{Krawtchouk polynomials case}]\label{krawtchouk}

As the Hamiltonian \eqref{H} we choose 
\begin{equation}\label{2pH}
H= (1-p)|z_0|^2 + p |z_1|^2 + \sqrt{p(1-p)} (z_0 \bar z_1 + \bar z_0 z_1 ) ,
\end{equation}
where $p\in (0,1)$ , and 
\begin{equation}
\rho = \frac{1}{2}\left(\begin{array}{cc} 
1 & -1 \\
1 & 1 
\end{array}\right) , 
\end{equation}
as the $\rho $ matrix. Thus, we find that
\begin{align}
I_0= \frac{1}{2} (|z_0|^2 -|z_1|^2), & \quad \psi_0 = \phi_0 -\phi_1, \\
I_1 = \frac{1}{2} (|z_0|^2 + |z_1|^2), & \quad \psi_1 = \phi_0 +  \phi_1 ,
\end{align}
where 
\begin{align}
\notag I_0+I_1 > 0 \\
\label{pte0} -I_0+I_1> 0, 
\end{align}
and  $0< \psi_1 \leq 4\pi $ and $-2\pi< \psi_0 \leq 2\pi$.

Expressing the Hamiltonian \eqref{2pH} in terms of the canonical coordinates $I_0, I_1, \psi_0, \psi_1$: 
\begin{equation}
H = (1-2p)I_0 + I_1 +2\sqrt{p(1-p)(I_1+I_0)(I_1-I_0)}\cos \psi_0
\end{equation}
and reducing system to $\textbf{J}^{-1}(c_1)/ \mathbb{T}$, $c_1 > 0$ we conclude that the reduces Hamilton equations are:
 \begin{align}
 \frac{dI_0}{dt} & =2 \sqrt{p(1-p)(c_1-I_0)(c_1+I_0)} \sin \psi_0 \\
 \frac{d\psi_0 }{dt} & =1-2p -  \frac{2\sqrt{p(1-p)}I_0}{\sqrt{(c_1-I_0)(c_1+I_0)}} \cos \psi_0    ,
 \end{align}
 where $-c_1 < I_0 <c_1 $. 

 The classical structural function $\mathcal{G}_0$ is the following
 \begin{equation}
 \mathcal{G}_0 (I_0, I_1) = p(1-p) (I_1+I_0)(I_1-I_0). 
 \end{equation}
 Thus, the Kummer shape is the circularly symmetric ellipsoid:
 \begin{equation}
 x^2 +y^2 + p(1-p)I_0^2 = p(1-p)c_1^2  
 \end{equation}
 and the classical 
 Kummer shape algebra is isomorphic to the Lie algebra of the group $SO(3)$:
 \begin{equation}
 \{I_0, x\} = -y ,
 \end{equation}
 \begin{equation}
 \{ I_0, y \} = x ,
 \end{equation}
 \begin{equation}
 \{ x, y \} = - p(1-p)I_0 . 
 \end{equation}

Writing \eqref{2pH} in terms of $x,y$ and $I_0$ one finds 
 \begin{equation}
H = (1-2p)I_0 + c_1 + 2x 
\end{equation}
and thus
 \begin{align}
 \frac{dI_0}{dt} &= 2y ,\\
 \frac{dx}{dt} &= -(1-2p) y ,\\
 \frac{dy}{dt} &= (1-2p) x - 2p(1-p) I_0 . 
 \end{align}
 So, similarly to the previous example, the Hamilton equations are linear.

 Quantum version of \eqref{2pH} and the quantum structural function $\mathcal{G}_\hbar $ are
 \begin{multline}\label{qHr2}
 \textbf{H} = (1-p)a_0^*a_0 + p a_1^* a_1 + \sqrt{p(1-p)}(a_0a_1^* + a_0^* a_1 ) =\\
= (1-2p)A_0 +A_1 +A+A^*   
 \end{multline}
 and
 \begin{equation}
 \mathcal{G}_\hbar (A_0, A_1) = p(1-p) (A_1 + A_0 +\hbar ) (A_1 - A_0),
 \end{equation}
 respectively.

 For the vacuum state $|0, v_1 \rangle $, where $c_1 = \frac{\hbar}{2} v_1$ and $v_1 \in \mathbb{N}$, see definition \eqref{vacuum}, the reduced Hilbert space $\mathcal{H}_{c_1}$ has dimension $v_1 +1$. The generators of the quantum Kummer shape algebra act on the elements $|n, v_1 -n \rangle $, $n=0,1,\ldots, n_1$, of the orthonormal basis of $\mathcal{H}_{c_1}$ in the following way
 \begin{align}
 \textbf{A}_0 |n, v_1 -n \rangle &= \hbar (n - \frac{v_1}{2} ) |n, v_1 - n \rangle ,\\
 \textbf{A} |n, v_1 -n \rangle &= \hbar \sqrt{p(1-p) n(v_1 - n +1)} |n -1 , v_1 - n+1 \rangle ,\\
 \textbf{A}^* |n, v_1 -n \rangle &= \hbar \sqrt{p(1-p) (n+1) (v_1 - n)} |n+1, v_1 -n-1 \rangle .
 \end{align}
and satisfy $SO(3)$ Lie algebra relations
\begin{align}
[\textbf{A}_0, \textbf{A} ] &= -\hbar \textbf{A},\\
[\textbf{A}_0, \textbf{A}^* ] &= \hbar \textbf{A}^*,\\
[\textbf{A}, \textbf{A}^* ] &= -2p(1-p)\hbar \textbf{A}_0 .
\end{align}

Applying the reduction procedure described in Section \ref{sec:4}, see formulas \eqref{nowystan} and \eqref{withfactor}, we find that
\begin{equation}
P_{c_1} |z_0, z_1 \rangle = \frac{z_1^{v_1}}{\sqrt{\hbar^{v_1}}} |z ; 0, v_1 \rangle ,
\end{equation}
where
\begin{equation}\label{2pstates}
|z ; 0,v_1 \rangle =\sum_{n=0}^{v_1} \frac{z^n }{\sqrt{(p(1-p))^nn! (v_1-n)!}} | n, v_1 -n \rangle 
\end{equation}
and the complex variable $z$ is given by $z= \sqrt{p(1-p)}z_0z_1^{-1}$. 
The resolution of identity \eqref{2pstates} with respect to the measure $d\mu_{c_1}  (\bar z , z) = \rho_{c_1} (|z|^2) d|z|^2 d\psi_0$ with the density 
\begin{equation}
\rho_{c_1} (|z|^2) = \frac{1}{2\pi p(1-p)} \frac{(v_1+1)!}{\left(\frac{|z|^2}{p(1-p)}+1\right)^{v_1 +2}} 
\end{equation}
obtained by the formula \eqref{eldifferent}. Note also that 
\begin{equation}
\langle z; c_1 | w; c_1 \rangle  = \frac{1}{v_1!} \left( 1+ \frac{\bar z w}{p(1-p)}\right)^{v_1}
\end{equation}
is the reproducing kernel for $ d\mu_{c_1}  (\bar z , z)$.

Heisenberg equations assume the form:
\begin{align}
 \frac{d}{dt}\textbf{A}_0 (t) &= 2\textbf{Y}(t) \\
 \frac{d}{dt} \textbf{X}(t) &= -(1-2p) \textbf{Y}(t)\\ 
 \frac{d}{dt}\textbf{Y}(t) &= (1-2p) \textbf{X}(t) - 2p(1-p) \textbf{A}_0 (t). 
 \end{align}

\end{example}

\begin{example}[$q$\textit{-deformed Heisenberg-Weyl algebra}]\label{qheisenberg}

Here we consider the one-mode case, i.e. $N=0$ and $l_0=1$ with quantum Hamiltonian 
\begin{equation}\label{pqH}
\textbf{H} = g(a_0^*a_0 ) a_0 + a_0^* g( a_0^* a_0), 
\end{equation}
where the function $g : [0, +\infty [ \to \mathbb{R}$ is given by 
\begin{equation}
g(\textbf{A}_0)^2  =  \frac{\hbar}{1-q^{\frac{\hbar}{\alpha }}} \cdot \frac{1-q^{\frac{\textbf{A}_0+\hbar}{\alpha }}}{\textbf{A}_0 + \hbar }, 
\end{equation}
for $0<q<1$, $\alpha > 0 $. Note that for $N=0$ and $l_0=1$ one has $A_0 = a_0^* a_0$.  Hence, the operator 
\begin{equation}
A= g(a_0^*a_0) a_0 
\end{equation}
defined in \eqref{ao} satisfy \eqref{rel1} and the relations 
\begin{align}
\label{st1}
 \textbf{A}^* \textbf{A} = \hbar \frac{1-q^{\frac{\textbf{A}_0}{\alpha}}}{1- q^{\frac{\hbar}{\alpha}}} =: \mathcal{G}_\hbar ( \textbf{A}_0- \hbar) , \\
\label{st2}
 \textbf{A}\textbf{A}^* = \hbar  \frac{1-q^{\frac{\textbf{A}_0+\hbar}{\alpha}}}{1- q^{\frac{\hbar}{\alpha}}} =: \mathcal{G}_\hbar ( \textbf{A}_0), 
\end{align}
equivalent to the relations 
\begin{align}
 \textbf{A}\textbf{A}^* - \textbf{A}^*\textbf{A}= \hbar q^{\frac{\textbf{A}_0}{\alpha}}, \\
\label{heisenberg} 
 \textbf{A}\textbf{A}^* - q^\frac{\hbar}{\alpha} \textbf{A}^*\textbf{A} = \hbar . 
\end{align}
The Hamilton operator expressed in terms of the operators $\textbf{A}$, $\textbf{A}^*$ and $\textbf{A}_0$ takes the form 
\begin{equation}\label{qpHred}
\textbf{H} = \textbf{A} + \textbf{A}^* . 
\end{equation}

Since in the limit $\alpha \to \infty$ we obtain Heisenberg-Weyl algebra, it is reasonable to consider the quantum Kummer shape algebra defined by \eqref{rel1}, \eqref{st1} and \eqref{st2} as the $q^{\frac{\hbar}{\alpha}}$-deformation of the Heisenberg- Weyl algebra.

One has $\|\textbf{A}\| = \left(\frac{\hbar}{1-q^{\frac{\hbar}{\alpha}}}\right)^{\frac{1}{2}}$ and thus, the coherent states 
\begin{equation}
|z \rangle = |0\rangle + \sum_{n=1}^\infty \frac{\left(\frac{z}{\sqrt{\hbar}}\right)^n }{\sqrt{[1]\ldots [n]}} |n\rangle 
\end{equation}
are defined for $z\in \mathbb{D} = \left\{ z\in \mathbb{C}:|z|<\left(\frac{\hbar}{1-q^{\frac{\hbar}{\alpha}}}\right)^\frac{1}{2}\right\}$, where $[n]:=\frac{1-q^{\frac{\hbar}{\alpha }n}}{1-q^{\frac{\hbar}{\alpha}}} .$  

One has the resolution of identity
\begin{equation}
\mathbbm{1} = \int_{\mathbb{D}} \frac{|z \rangle \langle z| }{\langle z | z \rangle } d\mu_\hbar (\bar z , z) , 
\end{equation}
where 
\begin{equation}
d\mu_\hbar (\bar z , z) = \frac{1}{1-(1-q^{\frac{\hbar}{\alpha}})\frac{\bar z z}{\hbar}} \cdot \frac{1}{2\pi} \sum_{n=0}^\infty q^{\frac{\hbar}{\alpha}n} \delta \left(|z|^2 - \frac{q^{\frac{\hbar}{\alpha}n}}{1-q^\frac{\hbar}{\alpha} }\right) d|z|^2 d\psi , 
\end{equation}
see \cite{OM}.

Now let us pass to the classical case. From \eqref{st1} we have 
\begin{equation}
\textbf{A}_0 = \frac{\alpha}{\ln q} \ln \left( 1- \frac{1-q^\frac{\hbar}{\alpha}}{\hbar} \textbf{A}^* \textbf{A}\right) = \frac{\alpha}{\ln q} \sum_{n=1}^\infty \frac{1}{n}\left(\frac{1-q^{\frac{\hbar}{\alpha}}}{\hbar}\right)^n (\textbf{A}^*\textbf{A})^n. 
\end{equation}
From 
\begin{equation}
\lim_{\hbar \to 0} \frac{\langle z | (\textbf{A}^*\textbf{A})^n |z \rangle }{\langle z |z\rangle } = (\bar z z)^n 
\end{equation}
we find 
\begin{equation}
\lim_{\hbar \to 0 } \langle \textbf{A}_0^n \rangle  = I_0^n , 
\end{equation}
where 
\begin{equation}
I_0 = \frac{\alpha}{\ln q} \ln (1+ \frac{\ln q}{\alpha } \bar z z). 
\end{equation}
So, the classical structural function is the following 
\begin{equation}
\mathcal{G}_0(I_0) = \frac{\alpha}{\ln q} (q^{\frac{I_0}{\alpha}}-1). 
\end{equation}

Classical Hamiltonian \eqref{H} takes the form
\begin{equation}\label{pqHclassical}
H= \sqrt{\frac{\alpha}{\ln q } \cdot \frac{q^{\frac{|z_0|^2}{\alpha}}-1}{|z_0|^2}} (z_0 + \bar z_0 ) = 2 \sqrt{\frac{\alpha}{\ln q}(q^{\frac{I_0}{\alpha}}-1)} \cos \psi_0
\end{equation}
and the corresponding Hamilton equations are 
\begin{align}
\frac{dI_0}{dt} & =  2 \sqrt{\frac{\alpha}{\ln q}(q^{\frac{I_0}{\alpha}}-1)} \sin \psi_0 , \\
\frac{d\psi_0}{dt} & = \frac{q^{\frac{I_0}{\alpha }}}{\sqrt{\frac{\alpha}{\ln q} (1- q^{\frac{I_0}{\alpha}})}}. 
\end{align}
In considered example the classical Kummer shape algebra is given by the relations
\begin{align}
\{I_0, x \} & = -y, \\
\{ I_0, y \} & = x , \\
\{ x, y \} & = \frac{1}{2} q^{\frac{I_0}{\alpha }} . 
\end{align}
Note that the Kummer shape 
\begin{equation}\label{pqkummershape}
x^2+y^2 = \frac{\alpha}{\ln q} (q^{\frac{I_0}{\alpha}}-1), 
\end{equation}
in the limit $q \to 1 $ takes the form 
\begin{equation}
x^2+y^2 = I_0,
\end{equation}
so, \eqref{pqkummershape} is a $q$-deformation of the elliptic paraboloid.

Hamilton equations corresponding to the Hamiltonian \eqref{pqHclassical} expressed in terms of $x,y$ and $I_0$ 
\begin{equation}
H= 2x, 
\end{equation}
are the following 
\begin{align}
\notag
\frac{dI_0}{dt} & = 2y, \\
\label{equationsclassical}
\frac{dx}{dt} & = 0, \\
\notag
\frac{dy}{dt} & = q^{\frac{I_0}{\alpha }}. 
\end{align}
One can easily verify that solutions of the Heisenberg equations
\begin{align}
\notag
\frac{d}{dt} \textbf{A}_0 (t) & = 2 \textbf{Y} (t) , \\
\frac{d}{dt} \textbf{X} (t) & = 0 , \\
\notag
\frac{d}{dt} \textbf{Y} (t) & = q^{\frac{\textbf{A}_0 (t) }{\alpha}} 
\end{align}
defined by Hamiltonian \eqref{qpHred}, in the limit $\hbar \to 0$ satisfy equations  \eqref{equationsclassical}. 
\end{example}

Let us stress that there is many other important examples which can be considered through the theory presented in this paper. The case $(l_0, l_1, l_2 ) = (1,-1,-1)$ and $g_0=const$, which describes non-linear interaction between three modes, plays an important role in non-linear optics. On the classical level it was considered in \cite{Luther}, \cite{Alber} and \cite{Alber1}, for its quantum counterpart see e.g. \cite{Jag2}. In order to see the importance of the quantum Kummer shape algebras in various problems of mathematical physics, besides the quantum optics, see e.g. \cite{Fern} and \cite{higgs}, \cite{Lemon}, where the applications of these algebras to factorization method and the Kepler problem on the $N$-sphere, respectively, were considered. 
$$\mbox{ } $$
$$\mbox{ } $$

\textbf{Acknowledgments:}\\
The authors would like to thank to M.Horowski for his interest and helpful remarks.

\end{document}